\documentclass[11pt,oneside]{amsart}

\usepackage{mathtools,upgreek,url,
}

\DeclareMathOperator{\Tr}{Tr}
\newcommand{\uv}{\mathrm{u}}
\newcommand{\du}{\mathrm{v}}
\newcommand{\muv}{\Tilde{\uv}}
\newcommand{\mdu}{\mathrm{w}}
\newcommand{\sop}{\mathbf{S}}
\newcommand{\tso}{\mathbf{T}}
\newcommand{\aly}{\upalpha}

\newcommand{\pp}{\mathrm{p}}
\newcommand{\qq}{\mathrm{q}}
\newcommand{\hh}{\mathrm{h}}

\newcommand{\cS}{\sop}
\newcommand{\tdt}{\mathbf{D}_t}

\newcommand{\ac}{\mathrm{b}}

\newcommand{\ilk}{\ell}
\newcommand{\iup}{i}
\newcommand{\nt}{n,t}
\newcommand{\km}{k}

\newcommand{\fde}{\mathbb{F}}

\newcommand{\ord}{\mathrm{p}}

\newcommand{\bC}{\mathbb{C}}

\newcommand{\mm}{\mathbf{M}}
\newcommand{\hmm}{\widehat{\mm}}

\newcommand{\md}{\mathbf{M}}
\newcommand{\hmd}{\widehat{\md}}

\newcommand{\kd}{\mathbf{K}}

\newcommand{\gbf}{\mathbf{G}}

\newcommand{\dd}[1]{\partial_{#1}}

\newcommand{\ws}[1]{\mathbb{W}_{#1}}
\newcommand{\tws}[1]{\widetilde{\mathbb{W}}_{#1}}

\newcommand{\hs}[1]{\mathbb{B}_{#1}}

\newcommand{\eo}[1]{\boldsymbol\Delta^{#1}}

\newcommand{\mc}{\mathrm{c}}

\newcommand{\ud}{\uv_\ilk,\ldots}

\newcommand{\una}{\upalpha}
\newcommand{\unb}{\upbeta}
\newcommand{\ua}{\una}
\newcommand{\ub}{\unb}

\newcommand{\tua}{{\tilde{\una}}}
\newcommand{\tub}{{\tilde{\unb}}}

\newcommand{\fn}{f}

\newcommand{\td}{w}

\newcommand{\mai}{\mathrm{b}}

\newcommand{\mf}{\Phi}

\newcommand{\la}{\uplambda}
\newcommand{\nb}{\nabla}
\newcommand{\tnb}{\widetilde{\nabla}}

\newcommand{\feq}{\mathrm{F}}

\newcommand{\tff}{\widetilde{\feq}}

\newcommand{\feqd}{\mathbf{F}}
\newcommand{\ixd}{\mathrm{m}}

\newcommand{\sm}{\mathrm{d}}

\newcommand{\ik}{\ell}

\newcommand{\lb}{\label}
\newcommand{\er}{\eqref}

\newcommand{\zp}{\mathbb{Z}_{\ge 0}}
\newcommand{\zsp}{\mathbb{Z}_{>0}}
\newcommand{\zz}{\mathbb{Z}}

\newcommand{\cl}{\colon}
\newcommand{\pd}{\partial}
\newcommand{\lwt}{\mathbf{W}}
\newcommand{\lt}{\lwt}

\newcommand{\diunw}{N}
\newcommand{\diu}{\diunw}

\newcommand{\pw}{\mathrm{q}}
\newcommand{\ki}{k}
\newcommand{\rpw}{r}
\newcommand{\fl}{f}
\newcommand{\pg}{g}
\newcommand{\ph}{h}

\newcommand{\tiltdt}{\widetilde{\tdt}}

\newcommand{\mer}{\mathcal{M\!E\!R}(\mathbb{C},\la)}
\newcommand{\uxd}{u_{\ixd x},\ldots}
\newcommand{\azc}{\mathbf{A}}
\newcommand{\bzc}{\mathbf{B}}

\newcommand{\ce}[1]{\mathbf{C}^{#1}}

\newtheorem{theorem}{Theorem}
\newtheorem{proposition}{Proposition}
\newtheorem{lemma}{Lemma}

\theoremstyle{definition}
\newtheorem{definition}{Definition}
\newtheorem{example}{Example}

\newtheorem{remark}{Remark}

\begin{document}

\title[The gauge action on semi-discrete Lax representations]%
{The gauge action on semi-discrete Lax representations, \\ 
its invariants, and Lax--Miura type constructions}
\date{}

\author{Sergei Igonin} 
\address{Center of Integrable Systems, P.G. Demidov Yaroslavl State University, Yaroslavl, Russia}
\email{s-igonin@yandex.ru}



\begin{abstract}

Semi-discrete (differential-difference) matrix Lax representations (Lax pairs)
play an essential role in the theory of (nonlinear) 
integrable differential-difference equations.
For a given equation, two matrix Lax representations (MLRs) 
are said to be gauge equivalent if one of them can be obtained from the other 
by applying a (local) matrix gauge transformation (GT).
GTs form an infinite-dimensional group, which acts on the set of MLRs of a given equation.
Two MLRs are gauge equivalent iff they belong to the same orbit of this action.
When one tries to establish integrability (in the sense of soliton theory)
for a given equation, one is interested in MLRs which depend on a parameter 
such that the parameter cannot be removed by any GT.

We introduce and study explicit 
invariants with respect to the action of GTs on the set of MLRs 
for a given (1+1)-dimensional evolutionary differential-difference equation 
with any number of components.
Using these invariants, we obtain the following general results:
\begin{itemize}
\item Sufficient conditions for nonremovability of parameters in semi-discrete MLRs.
\item Necessary conditions for two semi-discrete MLRs to be gauge equivalent.
\end{itemize}
Effective use of these results is demonstrated.
We present the general theory and apply it to MLRs for
equations of Volterra--Yamilov, Cherdantsev--Yamilov, Toda types
and Miura-type modifications.

Also, we present a comparison with results of S.Yu.~Sakovich and M.~Marvan
on gauge invariants for (continuous) zero-curvature representations 
of partial differential equations.

Furthermore, we demonstrate procedures for constructing 
Miura-type transformations (MTs) for discrete (partial difference) equations 
(on the two-dimensional integer lattice) by means of
gauge simplifications of discrete matrix Lax pairs (DLPs) 
and reductions of auxiliary systems associated with DLPs.
This allows us to construct new MTs 
for (non-evolutionary) 4-point and 6-point discrete equations 
from [D.~Levi, R.I.~Yamilov, \emph{J. Phys. A} 2009],
[S.~Konstantinou-Rizos, A.V.~Mikhailov, P.~Xenitidis, \emph{J. Math. Phys.} 2015].

\end{abstract}

\maketitle


\section{Introduction}
\subsection{An overview of the results}
\lb{sbint}

In this paper we study interrelations of (semi-)discrete matrix Lax representations (Lax pairs),
matrix gauge transformations, and discrete Miura-type transformations,
which belong to the main tools in the theories of (nonlinear) 
integrable differential-difference (semi-discrete) and integrable partial difference (discrete) equations.
Currently, such equations are explored very actively,
since they occupy a prominent place in the modern theory of integrable systems
and are widely used in physics 
(e.g. approximate descriptions of the dynamics of crystals and molecular 
systems~\cite{BenettinTodaMolecularDyn23,ChongKevrDynLattArxiv24,ChongGeislerKevrWaveMotion24,%
MutoTodaLatticeDNAPhysD1990,TodaBook1989}),
mathematical biology (e.g. integrable lattice models for food 
chains~\cite{ShinjoLotkaVolterraPreyPredatorCogMath15,ChenHuNonisospectLotkaVolterraAAM23,%
AdlerBogoyavLattCatalanRJMPh24}),
numerical analysis (including integrable (semi-)discretizations of nonlinear partial differential equations, 
see, e.g.,~\cite{Hirota1977,HJN-book2016,LWY-book2022,suris2003} and references therein).

In this paper, $\zsp$ and $\zp$ are the sets of positive and nonnegative integers, respectively.
Let $\sm\in\zsp$.
Fix a (1+1)-dimensional evolutionary differential-difference (semi-discrete) equation 
and consider semi-discrete $\sm\times\sm$ matrix Lax representations (MLRs) of this equation.
Two MLRs are said to be gauge equivalent if one of them can be obtained from the other 
by applying a (local) $\sm\times\sm$ matrix gauge transformation (GT).
GTs form an infinite-dimensional group, which acts on the set of MLRs of a given equation.
Two MLRs are gauge equivalent if and only if they belong to the same orbit of this action.
(See Definition~\ref{dmlpgtnew} for more details.)

When one tries to establish integrability (in the sense of soliton theory)
for a given equation, one is interested in MLRs which depend on a parameter 
(sometimes called the \emph{spectral parameter})
such that it cannot be removed by any GT.
Such a parameter is said to be \emph{nonremovable}.
(See Definition~\ref{dprem} on this matter.)

In Section~\ref{ssdmlr} we introduce and study explicit 
invariants with respect to the action of GTs on the set of MLRs 
for a given (1+1)-dimensional evolutionary differential-difference equation~\er{sdde}
with any number~$\diunw$ of components.
(As explained in Section~\ref{sgth}, equation~\er{sdde} is imposed
on an $\diunw$-component vector-function
$\uv=(\uv^1(n,t),\ldots,\uv^{\diunw}(n,t))$.)
First, in Theorem~\ref{thdet} we introduce the invariants given by formulas~\er{inv}, 
using the notation from Definitions~\ref{dmlpgtnew},~\ref{dwd}.
Subsequently, in Remark~\ref{rinv} based on Theorem~\ref{thdet2}
we show that each Miura-type transformation~\er{uvf} over equation~\er{sdde}
(in the sense of Definition~\ref{defmtt})
allows us to derive another collection of invariants given by~\er{inv2}.

Using these invariants, we obtain the following general results.

\noindent
\textbf{a) \,Sufficient conditions for nonremovability of a parameter in a semi-discrete MLR:}
\begin{itemize}
	\item Consider an MLR $\big(\mm(\nt,\ud,\la),\,\lwt(\nt,\ud,\la)\big)$
	with a parameter~$\la$ (as described in Definition~\ref{dmlpgtnew})
	for a given equation~\eqref{sdde}.
	Theorem~\ref{thpar} says that: \\
if at least one of the invariants~\er{inv} computed 
for this MLR depends nontrivially on~$\la$, then the parameter cannot be removed by any GT.


Note that, as shown in Remark~\ref{ri0}, we have also the following: 
if at least one of the invariants~\er{inv} computed for this MLR
is nonzero then the considered MLR is not gauge equivalent to a trivial MLR.
(The notion of trivial MLRs is understood in the sense of Definition~\ref{dtrmlp}.
Remark~\ref{ri0} uses Theorem~\ref{thtriv}, which is proved in~\cite{IgMLR26}.
See also Remark~\ref{rtdep} on the proof of Theorem~\ref{thtriv}.)
\item Consider an MLR $\big(\mm(\nt,\ud,\la),\,\lwt(\nt,\ud,\la)\big)$
with a parameter~$\la$ for equation~\eqref{sdde}.
Suppose that one has a Miura-type transformation~\er{uvf} 
from some equation~\er{vdde} to equation~\er{sdde}, in the sense of Definition~\ref{defmtt}.
Remark~\ref{rinv} and Theorem~\ref{thdet2} explain 
how to compute the invariants~\er{inv2} for this MLR, using the Miura-type transformation and other tools.
Theorem~\ref{thpar2} shows that: \\
if at least one of the computed invariants~\er{inv2}
depends nontrivially on~$\la$, then the parameter cannot be removed from the MLR by any GT.

Note also the following result in
Remark~\ref{ri0mt} based on Theorem~\ref{tmtriv}:
if at least one of the computed invariants~\er{inv2}
is nonzero then the considered MLR is not gauge equivalent to a trivial MLR.
\end{itemize}
\textbf{b) \,Necessary conditions for two semi-discrete MLRs to be gauge equivalent:}
\begin{itemize}
	\item When one has two different MLRs for a given equation,
we present necessary conditions for the two MLRs to be gauge equivalent.
This result is described in Remark~\ref{rnec} based on Theorems~\ref{thdet},~\ref{thdet2}.
\end{itemize}
Note that equations, MLRs, and GTs are allowed to be non-autonomous 
with explicit dependence on the independent variables~$n,\,t$.

Effective use of these results is demonstrated in Section~\ref{ssdmlr}, where
we present the general theory and apply it (in Examples~\ref{exy1}--\ref{extoda}) 
to MLRs of the following differential-difference equations:
\begin{itemize}
	\item The scalar equation~\er{exsh1}, which is a special case of~(V4)
in the Yamilov list of integrable Volterra-type differential-difference equations~\cite{yam2006}
and is a symmetry of the (partial difference) quad-graph equation~$\mathrm{H}1$ from 
the Adler--Bobenko--Suris list~\cite{abs2003,MWXTMP2011}.
\item The non-autonomous scalar equation~\er{eqms}, 
which is a special case of~\cite[eq.~(28)]{CherdYamil1995} by I.Yu.~Cherdantsev and R.I.~Yamilov.
Note that the MLR~\er{mlrms} for equation~\er{eqms} depends explicitly on~$n,\,t$.
	\item The $2$-component equation~\er{MMBW} derived by 
	E.~Mansfield, G.~Mar\'i Beffa, Jing~Ping~Wang~\cite[Section~5.1.1]{MMBW2013}
in a study of invariant time evolutions of polygons in the centro-affine plane 
(in the framework of discrete differential geometry).
	\item The Toda lattice in the Flaschka--Manakov coordinates~\cite{FlaschPhR74,Manak75}
	written as the $2$-component equation~\er{todafm}.
	\item The $2$-component equation~\er{ceqnew} from~\cite{IgMLR26}, 
	which is connected by the Miura-type transformations~\er{mtmtd} and~\er{mtcom} 
	to equations~\er{MMBW} and~\er{todafm}.
\end{itemize}
Examples~\ref{exy1}--\ref{extoda} illustrate the above general results.

In our opinion,
the significance of the above results is confirmed by the following observations:
\begin{itemize}
	\item 
	As said above, when one tries to establish integrability (in the sense of soliton theory)
for a given equation, one is interested in MLRs with a nonremovable parameter.
Furthermore, it is natural to study MLRs up to gauge equivalence.
Therefore, it is important to have sufficient conditions for nonremovability of a parameter 
and necessary conditions for two MLRs to be gauge equivalent, 
but, to our knowledge, the existing literature does not contain such conditions 
in the semi-discrete case.
\item Since semi-discrete MLRs and the action of GTs on semi-discrete MLRs
play pivotal roles in the theory of integrable differential-difference equations,
it is important to introduce and study explicit invariants with respect to this action, 
but we do not find such invariants in the literature.
(Note that, as discussed in Section~\ref{szcr}, 
some invariants for (continuous) zero-curvature representations of partial differential equations
are available in works of S.Yu.~Sakovich~\cite{sakov95,sakov2004} 
and M.~Marvan~\cite{marvan93,marvan97}.)
\item The theory of Section~\ref{ssdmlr} is developed in a quite general setting:
(1+1)-dimensional evolutionary differential-difference equations 
can have any number~$\diunw$ of components; MLRs and GTs can depend 
on any finite number of the shifts (dynamical variables) 
$\uv_\ik=(\uv^1_\ik,\ldots,\uv^{\diu}_\ik)$, $\,\ik\in\zz$,
of $\uv=(\uv^1,\ldots,\uv^{\diunw})$, on~$n,\,t$, and on a parameter~$\la$.
\end{itemize}

As discussed in Remark~\ref{rrat},
for simplicity of exposition, in order to present the main ideas clearly and concisely, 
in Section~\ref{ssdmlr} we consider functions which depend rationally on
$n$, $t$, $\uv_\ik=(\uv^1_\ik,\ldots,\uv^{\diu}_\ik)$ and meromorphically on~$\la$.
In particular, MLRs and GTs are given by matrix-valued functions of this type.
(The results of Section~\ref{ssdmlr} can be extended to more general classes of functions, 
but such extension requires more technical and heavy proofs.)

Some of our results on semi-discrete MLRs of differential-difference equations 
are inspired by results of S.Yu.~Sakovich~\cite{sakov95} and 
M.~Marvan~\cite{marvan93,marvan97} on 
(continuous) zero-curvature representations of partial differential equations.
A comparison with the works of S.Yu.~Sakovich 
and M.~Marvan is given in Section~\ref{szcr}.

It is well known that Miura-type transformations (MTs) are among the principal tools 
in the theories of integrable partial differential, partial difference 
and differential-difference equations. 
In particular, when one makes a classification for some class of such equations,
often it is possible to choose a few basic equations (with parameters)
such that all the other equations from the considered class can be derived from 
the basic ones by using~MTs 
(see, e.g.,~\cite{mss91,yam2006,garif2017,garif2018,meshk2008,%
LWY-book2022,SokBook2020,MikhNovJPWGlob26,BalakhCMMP23,gram2011} 
and references therein). 
MTs for partial differential equations are also called differential substitutions.

MTs for partial differential and differential-difference equations have 
already been studied quite extensively 
(see, e.g.,~\cite{MiuraI1968,drin-sok85,SokBook2020,LWY-book2022,mss91,SvinSokRMS1992,%
meshk2008,IgonMTJPA2005,MikhNovJPWGlob26,yam2006,HabibMT2023,%
garif2018,BIg2016,IgJGP2025,IgMLR26} and references therein), 
while the case of partial difference equations in this respect is much less understood.
Therefore, it is desirable to develop constructions of MTs for partial difference equations.

Below we use the following terminology:
an MT from an equation~$E_1$ to an equation~$E_2$ is called new if the whole triple 
$$
\big(\text{equation~$E_1$},\,\text{formula for the MT},\,\text{equation~$E_2$}\big) 
$$
(considered up to invertible changes of variables) 
has not been published before. 

In Section~\ref{sdiscr} we demonstrate procedures for constructing MTs 
for partial difference (discrete) equations on the two-dimensional integer lattice, 
using simplifications of discrete matrix Lax pairs (DLPs) by means of gauge transformations 
and reductions of auxiliary systems associated with DLPs.
(We call these procedures ``\emph{Lax--Miura type constructions}''.)
Namely, we do the following:
\begin{itemize}
	\item Let $\aly(m)$ be a rational scalar function of an integer variable~$m$.
	Using the notation explained in Section~\ref{sdiscr},
	in Example~\ref{exdgt} we consider the 4-point equation~\er{eqly}
	for a scalar function $\du=\du(n,m)$
	with the DLP~\er{dlply} from~\cite[Section~3]{LeYaJPA2009}
	by D.~Levi and R.I.~Yamilov.
	We	show how to find a GT which allows one to simplify this DLP 
	and bring it to a form suitable for constructing an MT, which is derived in Example~\ref{exdmt}.
	The resulting simplified form~\er{newdlp} of the DLP 
	has not been published before (as far as we are aware).
	\item In Example~\ref{exdmt}, using a reduction of the auxiliary system~\er{syspsit} 
	associated with the DLP~\er{newdlp}, 
	we construct the MT~\er{mtly} from~\er{eqmdu} to~\er{eqly}.
	In formulas~\er{syspsit}, \er{mtly}, \er{eqmdu} one has a constant $\mc\in\mathbb{C}$
	and a rational function~$\aly(m)$.
	In the non-autonomous case when the function $\aly(m)$ is nonconstant,
	$\mc$ can be arbitrary, and the constructed MT is new (to our knowledge).
	
	In the autonomous case when $\aly=\aly(m)$ is constant,
	we take~$\mc$ such that $\mc^2+4\aly^2\neq 0$.
	If $\aly$ is a nonzero constant and $\mc\neq 0$, then
	the MT~\er{mtly} from~\er{eqmdu} to~\er{eqly}
	can be found (up to an invertible change of variables) 
in~\cite{RamGramSatsWill2009} by A.~Ramani, B.~Grammaticos, J.~Satsuma, and R.~Willox.
If $\aly$ is a nonzero constant and $\mc=0$, 
	the MT	can be found in~\cite{gram2011} (also up to an invertible change of variables).
\item In Example~\ref{exkrmx} we consider the 6-point equation~\er{kmxp}
	with the DLP~\er{lkmx} from	~\cite[Section~IV]{KMX2015}
	by S.~Konstantinou-Rizos, A.V.~Mikhailov, and P.~Xenitidis.
	Using a reduction of the auxiliary system~\er{kmxsy} associated with the DLP~\er{lkmx}, 
	we construct the MT~\er{mtkmx1} from~\er{qqeq1} to~\er{kmxp}.	
	To our knowledge, this MT is new. 
	That is, the whole triple 
$\big(\text{eq.~\er{qqeq1}},\,\text{formula~\er{mtkmx1}},\,\text{eq.~\er{kmxp}}\big)$
(considered up to invertible changes of variables)
is not available in the existing literature, as far as we are aware.
\item In Example~\ref{exdlp}, 
substituting the MT~\er{mtkmx1} and its shifts into the matrices~\er{lkmx}, 
we derive the DLP~\er{dlpkm} for equation~\er{qqeq1}.
After that, we show how to find a GT which allows one to simplify the DLP~\er{dlpkm}.
The resulting simplified DLP is~\er{wovmk}.
\end{itemize}
DLPs and auxiliary systems associated with DLPs 
for partial difference equations on the two-dimensional lattice
are understood in the sense of Definition~\ref{ddmlr}.


In Remark~\ref{rcycl} we outline the following direction for future research:\\
to develop a theory of cyclic bases for 
semi-discrete MLRs of differential-difference equations
as a ``semi-discrete (differential-difference) analog''
of S.Yu.~Sakovich's theory of cyclic bases for (continuous) zero-curvature representations
of partial differential equations~\cite{sakov95,sakov2004,SakTrLax}.

\subsection{Conventions and notation}
\lb{sbnot}

The following conventions and notation are used in this paper.
\begin{itemize}
\item Unless otherwise specified, scalar functions take values in~$\bC$.
  \item Matrix-valued functions are sometimes called simply matrices.
	\item All considered matrices are of size $\sm\times\sm$ for some $\sm\in\zsp$.
	\item MLR = matrix Lax representation for a differential-difference (semi-discrete) equation.
	\item DLP = discrete matrix Lax pair for a discrete (partial difference) equation.
	\item ZCR = zero-curvature representation for a partial differential equation.
	\item GT = gauge transformation (given by an invertible matrix).
	\item MT = Miura-type transformation.
  \item For any variable~$v$, the symbol $\dd{v}$ denotes the partial derivative with respect to~$v$.
\end{itemize}

See also:
\begin{itemize}
	\item the conventions in Remark~\ref{rrat} about functions considered in Section~\ref{ssdmlr};
	\item the convention in Remark~\ref{ranf} on functions in Section~\ref{szcr}.
\end{itemize}

\section{Results on semi-discrete matrix Lax representations for differential-difference equations}
\lb{ssdmlr}

\subsection{The general theory}
\lb{sgth}

Let $\diunw\in\zsp$. Below we consider a  differential-difference equation 
for an $\diunw$-component vector-function 
$\uv=\big(\uv^1(n,t),\ldots,\uv^{\diunw}(n,t)\big)$
of an integer variable~$n$ and a real or complex variable~$t$.
For any fixed integer~$\ilk$, 
one has $\uv_\ilk=(\uv^1_\ilk,\ldots,\uv^{\diunw}_\ilk)$,
where for each $\iup=1,\ldots,\diunw$ 
the component $\uv_\ilk^\iup$ is a function
of~$n,t$ defined as follows $\uv_\ilk^\iup(n,t):=\uv^\iup(n+\ilk,t)$.
That is, $\uv_\ilk(n,t)=\uv(n+\ilk,t)$. In particular, $\uv_0=\uv$.

Let $\una,\unb\in\zz$, $\,\una\le\unb$.
Consider a (1+1)-dimensional evolutionary differential-difference (semi-discrete) equation
\begin{gather}
\lb{sdde}
\uv_t=\feq(\nt,\uv_\una,\uv_{\una+1},\ldots,\uv_\unb),
\end{gather}
where $\feq$ is an $\diunw$-component vector-function $\feq=(\feq^1,\ldots,\feq^{\diunw})$, 
which may depend on~$n$, $t$ and on~$\uv_\ilk$ for $\ilk=\una,\una\!+\!1,\ldots,\unb$.

The differential-difference equation \eqref{sdde} 
is equivalent to the following infinite set of differential equations
\begin{gather}
\lb{infcol}
\frac{d}{dt}\big(\uv(n,t)\big)=\feq\big(\nt,\uv(n+\una,t),\uv(n+\una+1,t),\ldots,\uv(n+\unb,t)\big),\qquad\quad n\in\zz.
\end{gather}
Using the components of $\uv=\big(\uv^1(n,t),\ldots,\uv^{\diunw}(n,t)\big)$ 
and of $\feq=(\feq^1,\ldots,\feq^{\diunw})$, we can rewrite~\eqref{sdde} as
\begin{gather}
\lb{msdde}
\uv^i_t=\feq^i(\nt,\uv_\una,\uv_{\una+1},\ldots,\uv_\unb),\qquad\quad
i=1,\ldots,\diunw.
\end{gather}
For each fixed $\ilk\in\zz$, replacing $n$ by $n+\ilk$ in~\er{infcol}, 
we see that \er{sdde} yields 
\begin{gather}
\lb{uildde}
\frac{d}{dt}\big(\uv^i_\ilk\big)=\feq^i(n+\ilk,t,\uv_{\una+\ilk},\uv_{\una+1+\ilk},\ldots,\uv_{\unb+\ilk}),
\qquad\quad
i=1,\ldots,\diunw,\qquad \ilk\in\zz.
\end{gather}

We use the formal theory of (1+1)-dimensional evolutionary differential-difference equations, where one regards 
\begin{gather}
\lb{uldiu}
n,\qquad t,\qquad
\uv_\ilk=(\uv^1_\ilk,\ldots,\uv^\diunw_\ilk),\quad\ilk\in\zz,
\end{gather}
as independent quantities.
Note that $\uv_\ilk=(\uv^1_\ilk,\ldots,\uv^\diunw_\ilk)$
are sometimes called \emph{dynamical variables}.
In what follows, the notation of the type $\fn=\fn(\nt,\ud)$ means 
that a function~$\fn$ depends on a finite number of the variables~\er{uldiu}.



We denote by~$\sop$ the \emph{shift operator} with respect to~$n$.
For any function $h=h(n,t)$, one has the function~$\sop(h)$
such that $\sop(h)(n,t)=h(n+1,t)$.
Furthermore, for each $k\in\zz$, we have the $k$th power~$\sop^k$ 
of the operator~$\sop$ so that $\sop^k(h)(n,t)=h(n+k,t)$.

Since $\uv_\ilk$ corresponds to $\uv(n+\ilk,t)$, the operator~$\sop$ 
and its powers~$\sop^k$ for $k\in\zz$ act on functions of
the variables~\er{uldiu} by means of the rules
\begin{gather}
\lb{csuf}
\begin{gathered}
\sop(n)=n+1,\qquad\sop(t)=t,\qquad\sop(\uv_\ilk)=\uv_{\ilk+1},\qquad
\sop^k(n)=n+k,\qquad\sop^k(\uv_\ilk)=\uv_{\ilk+k},\\
\sop^k\big(\fn(\nt,\uv_\ilk,\ldots)\big)=\fn(n+k,t,\sop^k(\uv_{\ilk}),\ldots).
\end{gathered}
\end{gather}
That is, applying $\sop^k$ to a function $\fn=\fn(\nt,\uv_\ilk,\ldots)$, 
we replace~$n$ by~$n+k$ in~$\fn$ 
and replace~$\uv_\ilk^{\iup}$ by~$\uv_{\ilk+k}^{\iup}$ in~$\fn$ for all $\ilk\in\zz$,
$\,\iup=1,\ldots,\diunw$.

\begin{remark}
\label{rrat}
In this section we deal with functions which may depend 
on a finite number of the variables~\eqref{uldiu} and a complex parameter~$\la$.
For simplicity of exposition, in order to present the main ideas clearly and concisely, 
we consider functions which depend rationally on
the variables~\eqref{uldiu} and meromorphically on~$\la$, in the sense described below.
The results of this section can be extended to more general classes of functions, 
but such extension requires more technical and heavy proofs.

A scalar function $\fn=\fn(\nt,\ud,\la)$ 
is said to be \emph{of admissible type} if $\fn$ can be written in the form
\begin{gather}
\label{fngh}
\fn=\frac{g(\nt,\ud,\la)}{h(\nt,\ud,\la)},
\end{gather}
where $g(\nt,\ud,\la)$ and $h(\nt,\ud,\la)$ are polynomials in the variables~\eqref{uldiu} 
whose coefficients are meromorphic functions of~$\la\in\bC$, 
and $h(\nt,\ud,\la)$ is not the zero polynomial.

We denote by~$\mer$ the field of meromorphic functions of~$\la\in\bC$.
Then~\eqref{fngh} can be regarded as a rational function of $n$, $t$, $\uv_\ik=(\uv^1_\ik,\ldots,\uv^{\diu}_\ik)$
with coefficients in~$\mer$.
Thus, a scalar function $\fn=\fn(\nt,\ud,\la)$ is of admissible type if $\fn$
belongs to the field of rational functions of $n$, $t$, $\uv_\ik$ with coefficients in~$\mer$.
We use the following convention:
in this section each scalar function of $n,\,t,\,\uv_\ell,\,\la$ is supposed to be of admissible type.
(In particular, according to this convention, 
the components $\feq^\iup=\feq^\iup(\nt,\uv_\una,\ldots,\uv_\unb)$, $\,\iup=1,\ldots,\diunw$,
of the vector-function $\feq=(\feq^1,\ldots,\feq^\diunw)$ from~\eqref{sdde} are rational.)

For any $\sm\in\zsp$, when we consider a $\sm\times\sm$ matrix-valued function
$\mathbf{P}=\mathbf{P}(\nt,\ud,\la)$, we assume that each entry of the matrix~$\mathbf{P}$ 
is a scalar function of admissible type.
When we say that such a matrix~$\mathbf{P}(\nt,\ud,\la)$ is supposed to be invertible, 
we mean that the determinant $\det\big(\mathbf{P}(\nt,\ud,\la)\big)$ is not identically zero. 
Then $\big(\det\big(\mathbf{P}(\nt,\ud,\la)\big)\big)^{-1}$ is well defined 
(as a rational function of $n$, $t$, $\uv_\ik$ with coefficients in~$\mer$), 
which allows us to compute the inverse matrix $\mathbf{P}^{-1}=\big(\mathbf{P}(\nt,\ud,\la)\big)^{-1}$, 
whose entries are of admissible type.

Starting from equation~\er{vdde} below, we consider also dynamical variables 
$\muv_\ell=\big(\muv^1_\ell,\dots,\muv^\diunw_\ell\big)$, $\,\ell\in\zz$.
The above convention about functions of $n,\,t,\,\uv_\ell,\,\la$ 
extends to functions of $n,\,t,\,\muv_\ell,\,\la$.
That is, each function of $n,\,t,\,\muv_\ell,\,\la$ is supposed to be 
a rational function of $n$, $t$, $\muv_\ell=\big(\muv^1_\ell,\dots,\muv^\diunw_\ell\big)$ with coefficients in~$\mer$.

\end{remark}

The \emph{total derivative operator~$\tdt$ corresponding to equation~\eqref{sdde}} 
acts on functions of $n,\,t,\,\uv_\ell,\,\la$ as follows
\begin{gather}
\lb{dtfunw}
\tdt\big(\fn(\nt,\ud,\la)\big)=\dd{t}(\fn)+\sum_{\substack{\iup=1,\ldots,\diunw,\\ \ilk\in\zz}}
\sop^\ilk(\feq^\iup)\cdot\dd{\uv_\ilk^\iup}(\fn),
\end{gather}
where $\feq^\iup=\feq^\iup(\nt,\uv_\una,\ldots,\uv_\unb)$ are the components of the vector-function 
$\feq=(\feq^1,\ldots,\feq^\diunw)$ from~\eqref{sdde}.
The sum in the right-hand side of~\er{dtfunw} is well defined,
since $\fn=\fn(\nt,\ud,\la)$ depends on a finite number of the variables~\er{uldiu} and~$\la$.
Formula~\er{dtfunw} reflects the chain rule for the derivative with respect to~$t$, 
taking into account~\er{uildde}.
From~\er{csuf},~\er{dtfunw} it follows that, for any function $g=g(\nt,\ud,\la)$,
we have $\tdt\big(\sop(g)\big)=\sop\big(\tdt(g)\big)$.

\begin{definition}
\lb{dmlpgtnew}
Let $\sm\in\zsp$. 
Let $\mm=\mm(\nt,\ud,\la)$ and $\lwt=\lwt(\nt,\ud,\la)$ be $\sm\times\sm$ matrix-valued functions.
Suppose that the matrix~$\mm$ is invertible (in the sense of Remark~\ref{rrat}) and one has
\begin{gather}
\lb{lrnw}
\tdt(\mm)=\sop(\lwt)\cdot\mm-\mm\cdot\lwt,
\end{gather}
where $\tdt$ is given by~\eqref{dtfunw}.
Then the pair $(\mm,\lwt)$ is called a 
\emph{semi-discrete $\sm\times\sm$ matrix Lax representation} (MLR) for equation~\eqref{sdde}.

The connection between~\eqref{lrnw} and~\eqref{sdde} is as follows.
The components~$\feq^\iup$, $\,i=1,\ldots,\diunw$, 
of the right-hand side $\feq(\nt,\uv_\una,\ldots,\uv_\unb)$ of~\eqref{sdde} 
appear in formula~\er{dtfunw} for the operator~$\tdt$, which is used in~\er{lrnw}.

Relation~\er{lrnw} implies that the following (overdetermined) auxiliary linear system
\begin{gather}
\notag
\sop(\Psi)=\mm\cdot\Psi,\qquad\quad
\pd_t(\Psi)=\lwt\cdot\Psi
\end{gather}
is compatible modulo equation~\eqref{sdde}. 
Here $\Psi=\Psi(n,t)$ is a $\sm\times\sm$ matrix-valued function.

A \emph{gauge transformation} (GT) is an invertible $\sm\times\sm$ 
matrix $\gbf=\gbf(\nt,\ud,\la)$.
If $(\mm,\lwt)$ is an MLR for~\eqref{sdde}, then the matrices 
\begin{gather}
\lb{ggtmtu}
\widehat\mm=\sop(\gbf)\cdot \mm\cdot\gbf^{-1},\qquad\quad
\widehat\lwt=\tdt(\gbf)\cdot\gbf^{-1}+\gbf\cdot\lwt\cdot\gbf^{-1}
\end{gather}
form an MLR for equation~\eqref{sdde} as well.
The MLR $\big(\widehat \mm,\,\widehat\lwt\big)$ is \emph{gauge equivalent} 
to the MLR $(\mm,\lwt)$, and one can say that
$\big(\widehat \mm,\,\widehat\lwt\big)$ is obtained from $(\mm,\lwt)$ 
by means of the gauge transformation~$\gbf$.

GTs~$\gbf=\gbf(\nt,\ud,\la)$ form an infinite-dimensional group 
with respect to the multiplication of invertible $\sm\times\sm$ matrices,
where invertibility is understood in the sense of Remark~\ref{rrat}.
Formulas~\er{ggtmtu} determine an action of the group of GTs
on the set of MLRs for a given equation~\eqref{sdde}.
Two MLRs are gauge equivalent if and only if they belong to the same 
orbit of this action, which is called the \emph{gauge action}.
\end{definition}

\begin{definition}
\lb{dtrmlp}
An MLR~$(\mm,\lwt)$ for equation~\er{sdde} is said to be \emph{trivial} 
if the matrix~$\mm$ does not depend on~$\uv_\ilk$ for any $\ilk\in\zz$.
Then \er{dtfunw} gives $\tdt(\mm)=\pd_t(\mm)$, and from~\er{lrnw} 
one obtains $\sop(\lwt)=\mm\cdot\lwt\cdot\mm^{-1}+\pd_t(\mm)\cdot\mm^{-1}$, 
which implies that $\lwt$ does not depend on~$\uv_\ilk$ either
(as can easily be shown using the rules~\er{csuf}).
Hence a trivial MLR does not provide any information about equation~\er{sdde}.
\end{definition}

Fix $\sm\in\zsp$. 
In this subsection, all considered matrices are of size $\sm\times\sm$.
Definition~\ref{dwd} is taken from~\cite{IgMLR26}.
(The paper~\cite{IgMLR26} assumes 
that matrix-valued functions do not depend on~$t$,
but this does not change the formulas in Definition~\ref{dwd}.)

\begin{definition}
\lb{dwd}
Let $\ws{\sm}$ be the space of all $\sm\times\sm$ matrix-valued functions of the form $Q=Q(\nt,\ud,\la)$.
Let $\hs{\sm}\subset\ws{\sm}$ be the subset of all invertible $\sm\times\sm$ matrix-valued functions $M=M(\nt,\ud,\la)$.
For each $M\in\hs{\sm}$, we consider the operator $\nb_{M}\cl\ws{\sm}\to\ws{\sm}$ given by the formula
\begin{gather}
\lb{nbmq}
\nb_{M}(Q)=\sop(M^{-1}QM).
\end{gather}
The operator $\nb_{M}$ is invertible, and one has
\begin{gather}
\lb{nbm1q}
\nb_{M}^{-1}(Q)=M\cdot\sop^{-1}(Q)\cdot M^{-1}.
\end{gather}
For each $k\in\zz$, we denote by $\nb_{M}^k$ the $k$th power 
of the operator~$\nb_{M}$. In particular, one has
\begin{gather*}
\nb_{M}^0(Q)=Q,\qquad\quad
\nb_{M}^2(Q)=\nb_{M}(\nb_{M}(Q))=\sop(M^{-1}\sop(M^{-1}QM)M),\\
\nb_{M}^{-2}(Q)=\nb_{M}^{-1}(\nb_{M}^{-1}(Q))=M\sop^{-1}(M\sop^{-1}(Q)M^{-1})M^{-1}.
\end{gather*}
Furthermore, for each $j=1,\ldots,\diunw$, we consider the map $\eo{\uv^j}\cl\hs{\sm}\to\ws{\sm}$ given by
\begin{gather}
\lb{foreo}
\eo{\uv^j}(M)=\sum_{\ilk\in\zz}\nb_M^{-\ilk}\big(\dd{\uv^j_\ilk}(M)\cdot M^{-1}\big),\qquad\quad
M=M(\nt,\ud,\la)\in\hs{\sm},\qquad
j=1,\ldots,\diunw.
\end{gather}
The sum in the right-hand side of~\er{foreo} is well defined,
since $M=M(\nt,\ud,\la)$ depends on a finite number of the variables~\er{uldiu} and~$\la$.

For example, 
if $M=M(\nt,\uv^i_{-1},\uv^i_{0},\uv^i_{1},\la)$ depends only on~$n$, $t$, $\uv^i_{-1}$, 
$\uv^i_{0}$, $\uv^i_{1}$, $\,i=1,\ldots,\diunw$,  and on~$\la$, then 
\begin{multline*}
\eo{\uv^j}\big(M(\nt,\uv^i_{-1},\uv^i_{0},\uv^i_{1},\la)\big)=
\sum_{\ilk\in\zz}\nb_M^{-\ilk}\big(\dd{\uv^j_\ilk}(M)\cdot M^{-1}\big)=
\sum_{\ilk=-1,0,1}\nb_M^{-\ilk}\big(\dd{\uv^j_\ilk}(M)\cdot M^{-1}\big)=\\
=
\sop\big(M^{-1}\big(\dd{\uv^j_{-1}}(M)\cdot M^{-1}\big)M\big)
+\dd{\uv^j_0}(M)\cdot M^{-1}+M\sop^{-1}\big(\dd{\uv^j_1}(M)\cdot M^{-1}\big)M^{-1}.
\end{multline*}
\end{definition}
\begin{remark}
\label{rtdep}
Lemma~\ref{lem1} and Theorem~\ref{thtriv} below are proved in~\cite{IgMLR26} 
for matrix-valued functions and MLRs which do not depend on~$t$.
The proofs of these results in~\cite{IgMLR26} work also 
in the general case where dependence on~$t$ is allowed.
\end{remark}
\begin{lemma}[\cite{IgMLR26}]
\label{lem1}
Let $M=M(\nt,\ud,\la)$, $\,\gbf=\gbf(\nt,\ud,\la)$ 
be invertible $\sm\times\sm$ matrix-valued functions 
and $P=P(\nt,\ud,\la)$ be an arbitrary $\sm\times\sm$ matrix-valued function.
Set $\widetilde{M}=\sop(\gbf)\cdot M\cdot\gbf^{-1}$.
Then 
\begin{gather}
\lb{tmm}
\nb_{\widetilde{M}}^{\km}\big(\sop(\gbf)\cdot P\cdot \sop(\gbf)^{-1}\big)=
\sop(\gbf)\cdot \nb^{\km}_M(P)\cdot \sop(\gbf)^{-1}\qquad\quad
\forall\,{\km}\in\mathbb{Z}.
\end{gather}
\end{lemma}

The statement of Theorem~\ref{th1} is presented in~\cite{IgMLR26} 
with a sketch of proof in the case when $M$, $G$ do not depend on~$t$.
Here we give a detailed proof for this theorem in the general case 
allowing dependence on $n,\,t,\,\uv_\ell,\,\la$.

\begin{theorem}
\label{th1}
For any invertible $\sm\times\sm$ matrix-valued functions
\begin{gather}
\notag
M=M(\nt,\ud,\la),\qquad\quad G=G(\nt,\ud,\la),
\end{gather}
we have
\begin{gather}
\label{gmg}
\eo{\uv^j}\big(\sop(G)\cdot M\cdot G^{-1}\big)=\sop(G)\cdot\eo{\uv^j}(M)\cdot\sop(G)^{-1}
\qquad\quad\forall\,j=1,\ldots,\diunw.
\end{gather}
\end{theorem}
\begin{proof}
Let $\widetilde{M}=\sop(G)MG^{-1}$.
To prove~\er{gmg}, we need to show that 
$\eo{\uv^j}\big(\widetilde{M}\big)=\sop(G)\cdot\eo{\uv^j}(M)\cdot\sop(G)^{-1}$.

Using the relations
\begin{gather*}
\widetilde{M}^{-1}=G M^{-1}\sop(G)^{-1},\qquad 
\dd{\uv^j_\ilk}(G^{-1})=-G^{-1}\dd{\uv^j_\ilk}(G)G^{-1},\qquad
\dd{\uv^j_\ilk}\big(\sop(G)\big)=\sop\big(\dd{\uv^j_{\ilk-1}}(G)\big),\\
\sop\big(\dd{\uv^j_{\ilk-1}}(G)\big)=\sop(G)\sop\big(G^{-1}\dd{\uv^j_{\ilk-1}}(G)\big),\qquad
MG^{-1}\dd{\uv^j_\ilk}(G)M^{-1}=\nb_{M}^{-1}\big(\sop\big(G^{-1}\dd{\uv^j_\ilk}(G)\big)\big),
\end{gather*}
we get
\begin{multline}
\lb{tmm1}
\dd{\uv^j_\ilk}\big(\widetilde{M}\big)\cdot\widetilde{M}^{-1}=\dd{\uv^j_\ilk}\big(\sop(G) MG^{-1}\big)
G M^{-1}\sop(G)^{-1}=\\
=\Big(\dd{\uv^j_\ilk}\big(\sop(G)\big)MG^{-1}
+\sop(G)\dd{\uv^j_\ilk}(M)G^{-1}+
\sop(G) M\dd{\uv^j_\ilk}(G^{-1})\Big)G M^{-1}\sop(G)^{-1}=\\
=\Big(\dd{\uv^j_\ilk}\big(\sop(G)\big)MG^{-1}G M^{-1}+
\sop(G)\dd{\uv^j_\ilk}(M)G^{-1}G M^{-1}
-\sop(G) MG^{-1}\dd{\uv^j_\ilk}(G)G^{-1}
G M^{-1}\Big)\sop(G)^{-1}=\\
=\Big(\sop\big(\dd{\uv^j_{\ilk-1}}(G)\big)+\sop(G)\dd{\uv^j_\ilk}(M)M^{-1}
-\sop(G)MG^{-1}\dd{\uv^j_\ilk}(G)M^{-1}\Big)\sop(G)^{-1}=\\
=\sop(G)\Big(\sop\big(G^{-1}\dd{\uv^j_{\ilk-1}}(G)\big)+\dd{\uv^j_\ilk}(M)M^{-1}
-\nb_{M}^{-1}\big(\sop\big(G^{-1}\dd{\uv^j_\ilk}(G)\big)\big)\Big)\sop(G)^{-1}.
\end{multline}
Using~\er{tmm1} and Lemma~\ref{lem1}, one derives
\begin{multline*}
\eo{\uv^j}\big(\widetilde{M}\big)=
\sum_{\ilk\in\zz}\nb_{\widetilde{M}}^{-\ilk}\Big(\dd{\uv^j_\ilk}\big(\widetilde{M}\big)\cdot\widetilde{M}^{-1}\Big)=\\
=\sum_{\ilk\in\zz}\nb_{\widetilde{M}}^{-\ilk}\Big(\sop(G)\Big(\sop\big(G^{-1}\dd{\uv^j_{\ilk-1}}(G)\big)+
\dd{\uv^j_\ilk}(M)M^{-1}
-\nb_{M}^{-1}\big(\sop\big(G^{-1}\dd{\uv^j_\ilk}(G)\big)\big)\Big)\sop(G)^{-1}\Big)=\\
=\sop(G)\Big(\sum_{\ilk\in\zz}\nb_M^{-\ilk}\Big(\sop\big(G^{-1}\dd{\uv^j_{\ilk-1}}(G)\big)+
\dd{\uv^j_\ilk}(M)M^{-1}
-\nb_{M}^{-1}\big(\sop\big(G^{-1}\dd{\uv^j_\ilk}(G)\big)\big)\Big)\Big)\sop(G)^{-1}=\\
=\sop(G)\Big(\sum_{\ilk\in\zz}\nb_M^{-\ilk}\big(\dd{\uv^j_\ilk}(M)M^{-1}\big)+
\sum_{\ilk\in\zz}\nb_M^{-\ilk}\big(\sop\big(G^{-1}\dd{\uv^j_{\ilk-1}}(G)\big)\big)
-\sum_{\ilk\in\zz}\nb_M^{-(\ilk+1)}\big(\sop\big(G^{-1}\dd{\uv^j_\ilk}(G)\big)\big)\Big)\sop(G)^{-1}=\\
=\sop(G)\Big(\sum_{\ilk\in\zz}\nb_M^{-\ilk}\big(\dd{\uv^j_\ilk}(M)M^{-1}\big)\Big)\sop(G)^{-1}=
\sop(G)\cdot\eo{\uv^j}(M)\cdot\sop(G)^{-1}.
\end{multline*}
\end{proof}
\begin{theorem}[\cite{IgMLR26}]
\label{thtriv}
An MLR $\big(\mm(\nt,\ud,\la),\,\lwt(\nt,\ud,\la)\big)$ 
is gauge equivalent to a trivial MLR if and only if\/ 
$\mm=\mm(\nt,\ud,\la)$ satisfies $\eo{\uv^j}(\mm)=0$ for all $j=1,\ldots,\diunw$.
\end{theorem}

The statement of Lemma~\ref{lab} is a simple result from linear algebra.
\begin{lemma}
\lb{lab}
Let $A$, $B$ be $\sm\times\sm$ matrices.
Suppose that there is an invertible $\sm\times\sm$ matrix~$C$ such that $B=CAC^{-1}$.
Then 
\begin{gather}
\lb{dettr}
\det(B)=\det(A),\qquad\quad\Tr(B^\pw)=\Tr(A^\pw),\qquad\pw=1,\ldots,\sm-1.
\end{gather}
\end{lemma}
In what follows, unless otherwise specified, we consider $\sm\times\sm$ MLRs for a given equation~\er{sdde}.
\begin{theorem}
\lb{thdet}
For each MLR $\big(\mm(\nt,\ud,\la),\,\lwt(\nt,\ud,\la)\big)$, we consider
\begin{gather}
\label{inv}
\det\big(\eo{\uv^j}(\mm)\big),\qquad
\Tr\big(\big(\eo{\uv^j}(\mm)\big)^\pw\big),\qquad\quad
j=1,\ldots,\diunw,\qquad\pw=1,\ldots,\sm-1.
\end{gather}
The functions~\er{inv} are invariant with respect to the action 
of GTs on the set of MLRs for a given equation~\er{sdde},
in the following sense. 
Consider an MLR $\big(\mm(\nt,\ud,\la),\,\lwt(\nt,\ud,\la)\big)$, 
a  gauge transformation $\gbf=\gbf(\nt,\ud,\la)$, and the corresponding gauge equivalent MLR
\begin{gather}
\notag
\hmm=\sop(\gbf)\cdot\mm\cdot\gbf^{-1},\qquad\quad
\widehat\lwt=\tdt(\gbf)\cdot\gbf^{-1}+\gbf\cdot\lwt\cdot\gbf^{-1}.
\end{gather}
Then we have 
\begin{gather}
\lb{dettr1}
\begin{gathered}
\det\big(\eo{\uv^j}\big(\hmm\big)\big)=\det\big(\eo{\uv^j}(\mm)\big),\qquad\quad
\Tr\big(\big(\eo{\uv^j}\big(\hmm\big)\big)^\pw\big)=\Tr\big(\big(\eo{\uv^j}(\mm)\big)^\pw\big),\\
j=1,\ldots,\diunw,\qquad\quad\pw=1,\ldots,\sm-1.
\end{gathered}
\end{gather}
\end{theorem}
\begin{proof}
Since $\widehat\mm=\sop(\gbf)\cdot\mm\cdot\gbf^{-1}$, 
property~\er{gmg} from Theorem~\ref{th1} gives 
\begin{gather}
\notag
\eo{\uv^j}\big(\widehat{\mm}\big)=\sop(\gbf)\cdot\eo{\uv^j}(\mm)\cdot\sop(\gbf)^{-1}
\qquad\quad\forall\,j=1,\ldots,\diunw,
\end{gather}
which allows us to apply Lemma~\ref{lab} to
\begin{gather}
\notag
A=\eo{\uv^j}(\mm),\qquad\quad
B=\eo{\uv^j}\big(\widehat{\mm}\big),\qquad\quad C=\sop(\gbf).
\end{gather}
Then \er{dettr} gives~\er{dettr1}.
\end{proof}
\begin{remark}
\lb{ri0}
Consider an MLR $\big(\mm(\nt,\ud,\la),\,\lwt(\nt,\ud,\la)\big)$.
Suppose that at least one of the invariants~\er{inv} computed for this MLR is nonzero. 
Then we have $\eo{\uv^j}(\mm)\neq 0$ for some~$j$, 
and Theorem~\ref{thtriv} implies that this MLR is not gauge equivalent to a trivial MLR.
\end{remark}

In what follows, for any function $w=w(n,t)$ and each $\ell\in\zz$, 
we denote by~$w_\ell$ the function $w_\ell(n,t)=w(n+\ell,t)$.
That is, $w_\ell=\sop^\ell(w)$. In particular, $w_0=w$.

Now let $\tua,\tub\in\zz$, $\,\tua\le\tub$, 
and consider another (1+1)-dimensional evolutionary differential-difference equation
\begin{gather}
\lb{vdde}
\muv_t=\tff(\nt,\muv_\tua,\muv_{\tua+1},\ldots,\muv_\tub)
\end{gather}
for an $\diunw$-component vector-function $\muv=\big(\muv^1(n,t),\dots,\muv^{\diunw}(n,t)\big)$.
In the right-hand side of~\er{vdde} one has
an $\diunw$-component vector-function $\tff=\big(\tff^1,\dots,\tff^{\diunw}\big)$
depending (rationally) on $n,\,t,\,\muv_\tua,\,\muv_{\tua+1},\,\ldots,\,\muv_\tub$.

\begin{definition}
\lb{defmtt}
A \emph{Miura-type transformation} (MT) from equation~\eqref{vdde} 
to equation~\eqref{sdde} is determined by an expression of the form
\begin{gather}
\lb{uvf}
\uv=\mf(\nt,\muv_\ell,\dots)
\end{gather}
with the following requirements:
\begin{enumerate}
	\item The right-hand side of~\er{uvf} 
	is an $\diunw$-component (rational) vector-function $\mf=(\mf^1,\dots,\mf^{\diunw})$ 
	which  may depend on~$n$, $t$ and on a finite number of the dynamical variables 
	$\muv_\ell=\big(\muv^1_\ell,\dots,\muv^\diunw_\ell\big)$, $\,\ell\in\zz$, 
	and obeys the nondegeneracy condition described below.
	\item If $\muv=\muv(n,t)$ satisfies equation~\eqref{vdde} 
	then $\uv=\uv(n,t)$ determined by~\eqref{uvf} satisfies equation~\eqref{sdde}.
\end{enumerate}
More precisely, the second requirement means that we must have relations~\er{dtmfi} explained below.
The components of the vector formula~\eqref{uvf} are 
\begin{gather}
\lb{uimfi}
\uv^i=\mf^i\big(\nt,\muv_\ell,\dots\big),\qquad\quad i=1,\dots,\diunw.
\end{gather}
Recall that \er{sdde} is equivalent to~\er{msdde}, 
and we have $\uv^i_\ilk=\sop^\ilk(\uv^i)$ for all $i=1,\ldots,\diunw$ and $\ilk\in\zz$.
Substituting the right-hand side of~\eqref{uimfi} in place of~$\uv^i$ in~\eqref{msdde}, we get 
\begin{gather}
\lb{dtmfi}
\tiltdt\big(\mf^i\big(\nt,\muv_\ell,\dots\big)\big)=
\feq^i\big(\nt,\sop^\ua(\mf),\sop^{\ua+1}(\mf),\dots,\sop^\ub(\mf)\big),\qquad\quad
i=1,\dots,\diunw,
\end{gather}
which must be valid identically in the variables $n$, $t$ 
and $\muv_\ell=\big(\muv^1_\ell,\dots,\muv^\diunw_\ell\big)$.
Here $\tiltdt$ is the total derivative operator corresponding to~\eqref{vdde}, 
so that in the left-hand side of~\er{dtmfi} we have 
$$
\tiltdt\big(\mf^i\big(\nt,\muv_\ell,\dots\big)\big)=\pd_t(\mf^i)+
\sum_{\ell,\xi}\sop^\ell\big(\tff^\xi\big)\cdot\dd{\muv^\xi_\ell}(\mf^i),
$$
where $\tff^\xi=\tff^\xi\big(\nt,\muv_\tua,\dots,\muv_\tub\big)$,
$\,\xi=1,\ldots,\diunw$, are the components of the vector-function 
$\tff=\big(\tff^1,\dots,\tff^{\diunw}\big)$ from~\eqref{vdde}.

Furthermore, we impose the following \textbf{nondegeneracy condition}: \\
for any nonzero polynomial $p(\nt,\ud)$ with coefficients in~$\mathbb{C}$, the corresponding function
\begin{gather}
\lb{pbar}
\overline{p}(\nt,\muv_k,\ldots):=p(\nt,\cS^\ik(\mf),\ldots)
\end{gather}
is nonzero as a rational function of 
$n$, $t$ and $\muv_k=\big(\muv^1_k,\dots,\muv^\diunw_k\big)$, $\,k\in\zz$.
Here \er{pbar} is obtained from $p(\nt,\ud)$ 
by substituting $\uv_\ik=\cS^\ik(\mf)$ for $\ik\in\zz$.
(This substitution is motivated by formula~\eqref{uvf}.)

Having an MT~\er{uvf} from~\eqref{vdde} to~\eqref{sdde},
one can say that equation~\eqref{vdde} 
is a \emph{Miura-type modification} of equation~\eqref{sdde}.
We say also that the MT~\er{uvf} \emph{is over equation~\er{sdde}}.
Note that \er{uvf} does not depend on~$\la$.
\end{definition}
Recall that we denote by~$\mer$ the field of meromorphic functions of~$\la\in\bC$.
Recall also our convention from Remark~\ref{rrat}: in this section, 
each function of $n,\,t,\,\uv_\ell,\,\la$ is supposed to be 
a rational function of $n$, $t$, $\uv_\ell=\big(\uv^1_\ell,\dots,\uv^\diunw_\ell\big)$ with coefficients in~$\mer$, 
and similarly for functions of $n,\,t,\,\muv_\ell,\,\la$.
\begin{lemma}
\label{lcla}
Consider an MT~\er{uvf} from equation~\er{vdde} to equation~\er{sdde}.
For any polynomial $a(\nt,\ud,\la)$ in the variables~$n$, $t$, $\uv_\ik=(\uv^1_\ik,\ldots,\uv^{\diu}_\ik)$ 
with coefficients in~$\mer$, we denote by $\overline{a}(\nt,\muv_k,\ldots,\la)$
the function obtained from $a(\nt,\ud,\la)$ by substituting $\uv_\ik=\cS^\ik(\mf)$ for $\ik\in\zz$, 
where $\mf$ is the right-hand side of~\er{uvf}. 
That is, $\overline{a}(\nt,\muv_k,\ldots,\la):=a(\nt,\cS^\ik(\mf),\ldots,\la)$.

Let $\pg(\nt,\ud,\la)$, $\,\ph(\nt,\ud,\la)$, $\,a(\nt,\ud,\la)$, $\,b(\nt,\ud,\la)$ 
be polynomials in $n$, $t$, $\uv_\ik$ with coefficients in~$\mer$ such that 
$\ph(\nt,\ud,\la)\neq 0$, $\,b(\nt,\ud,\la)\neq 0$, and one has
\begin{gather}
\lb{fgab}
\frac{\pg(\nt,\ud,\la)}{\ph(\nt,\ud,\la)}=\frac{a(\nt,\ud,\la)}{b(\nt,\ud,\la)}
\end{gather}
in the field of rational functions of $n$, $t$, $\uv_\ik$ with coefficients in~$\mer$.

Then for the corresponding functions 
$\overline{\pg}(\nt,\muv_k,\ldots,\la)$, $\,\overline{\ph}(\nt,\muv_k,\ldots,\la)$,
$\,\overline{a}(\nt,\muv_k,\ldots,\la)$, $\,\overline{b}(\nt,\muv_k,\ldots,\la)$ 
we have $\overline{\ph}(\nt,\muv_k,\ldots,\la)\neq 0$, 
$\,\overline{b}(\nt,\muv_k,\ldots,\la)\neq 0$ and 
\begin{gather}
\lb{ofga}
\frac{\overline{\pg}(\nt,\muv_k,\ldots,\la)}{\overline{\ph}(\nt,\muv_k,\ldots,\la)}=
\frac{\overline{a}(\nt,\muv_k,\ldots,\la)}{\overline{b}(\nt,\muv_k,\ldots,\la)}
\end{gather}
 in the field of rational functions of $n$, $t$, $\muv_k$ with coefficients in~$\mer$.
\end{lemma}
\begin{proof}
Suppose that the function 
\begin{gather}
\lb{overg}
\overline{\ph}(\nt,\muv_k,\ldots,\la)=\ph(\nt,\cS^\ik(\mf),\ldots,\la)
\end{gather}
is zero.
Since $\ph(\nt,\ud,\la)$ is a polynomial in~$n$, $t$, $\uv_\ik$ 
with coefficients in~$\mer$, 
we can consider the (finite) set of the coefficients of this polynomial.
We denote this set by $\mathrm{Coeff}\subset\mer$. 

As $\ph(\nt,\ud,\la)$ is nonzero, there is a nonzero function $y(\la)\in\mathrm{Coeff}$.
Since the set $\mathrm{Coeff}$ is finite, there is a constant $c\in\bC$ such that 
$c$ does not belong to the sets of poles of the functions from~$\mathrm{Coeff}$, 
and one has $y(c)\neq 0$. 
Then we can substitute $c$ in place of~$\la$ in $\ph(\nt,\ud,\la)$, and the resulting
polynomial $\ph(\nt,\ud,c)$ with coefficients in~$\bC$ is nonzero.

Consider the rational function
\begin{gather}
\lb{ovgc}
\overline{\ph}(\nt,\muv_k,\ldots,c):=\ph(\nt,\cS^\ik(\mf),\ldots,c)
\end{gather}
obtained from $\ph(\nt,\ud,c)$ by substituting $\uv_\ik=\cS^\ik(\mf)$ for $\ik\in\zz$.
Since $\ph(\nt,\ud,c)$ is nonzero and 
the MT~\er{uvf} obeys the nondegeneracy condition from Definition~\ref{defmtt}, 
the function~\er{ovgc} is nonzero. This contradicts our assumption that 
$\overline{\ph}(\nt,\muv_k,\ldots,\la)$ given by~\er{overg} is zero.
By the same arguments, the fact that $b(\nt,\ud,\la)\neq 0$ yields
$\overline{b}(\nt,\muv_k,\ldots,\la)\neq 0$.

Since $\ph(\nt,\ud,\la)\neq 0$ and $b(\nt,\ud,\la)\neq 0$,
equation~\er{fgab} is equivalent to
\begin{gather}
\lb{pfgb}
\pg(\nt,\ud,\la)\cdot b(\nt,\ud,\la)=a(\nt,\ud,\la)\cdot \ph(\nt,\ud,\la).
\end{gather}
Substituting $\uv_\ik=\cS^\ik(\mf)$ for $\ik\in\zz$ into~\er{pfgb}, we derive
\begin{gather}
\lb{pofg}
\overline{\pg}(\nt,\muv_k,\ldots,\la)\cdot\overline{b}(\nt,\muv_k,\ldots,\la)=
\overline{a}(\nt,\muv_k,\ldots,\la)\cdot\overline{\ph}(\nt,\muv_k,\ldots,\la).
\end{gather}
Since $\overline{\ph}(\nt,\muv_k,\ldots,\la)\neq 0$ and $\overline{b}(\nt,\muv_k,\ldots,\la)\neq 0$,
equation~\er{pofg} is equivalent to~\er{ofga}.
\end{proof}

\begin{definition}
\lb{rlift}
Suppose that one has an MT~\er{uvf} 
from equation~\er{vdde} to equation~\er{sdde}.
Let $\fl=\fl(\nt,\ud,\la)$ be a scalar function.
According to the convention from Remark~\ref{rrat}, 
it belongs to the field of rational functions of $n$, $t$, $u_\ik$ with coefficients in~$\mer$.
That is, we have
\begin{gather}
\lb{rh}
\fl(\nt,\ud,\la)=\frac{\pg(\nt,\ud,\la)}{\ph(\nt,\ud,\la)},
\end{gather}
where $\pg(\nt,\ud,\la)$, $\,\ph(\nt,\ud,\la)$ are polynomials in
the variables~$n$, $t$, $\uv_\ik=(\uv^1_\ik,\ldots,\uv^{\diu}_\ik)$ with coefficients in~$\mer$, 
and $\ph(\nt,\ud,\la)$ is nonzero.

Taking into account Lemma~\ref{lcla}, we consider the function
\begin{gather}
\notag
\overline{\fl}(\nt,\muv_k,\ldots,\la):=
\fl(\nt,\cS^\ik(\mf),\ldots,\la)=
\frac{\pg(\nt,\cS^\ik(\mf),\ldots,\la)}{\ph(\nt,\cS^\ik(\mf),\ldots,\la)}
\end{gather}
obtained from~\er{rh} by substituting $\uv_\ik=\cS^\ik(\mf)$ for $\ik\in\zz$, 
where $\mf$ is the right-hand side of~\er{uvf}. 
We call $\overline{\fl}(\nt,\muv_k,\ldots,\la)$ 
the \emph{lifting} of the function~\er{rh} by means of the MT~\er{uvf}.
Lemma~\ref{lcla} implies that this operation of lifting is a well-defined map from
the field of rational functions of $n$, $t$, $\uv_\ik$ with coefficients in~$\mer$
to the field of rational functions of $n$, $t$, $\muv_k$ with coefficients in~$\mer$.

Now consider a matrix of the form $\mathbf{P}=\mathbf{P}(\nt,\ud,\la)$.
Denote by~$\overline{\mathbf{P}}$ the matrix obtained from~$\mathbf{P}$ 
by substituting $\uv_\ik=\cS^\ik(\mf)$ for $\ik\in\zz$.
We call~$\overline{\mathbf{P}}$ the \emph{lifting} of~$\mathbf{P}$ 
by means of the MT~\er{uvf}.
That is, to obtain~$\overline{\mathbf{P}}$, we apply the operation of lifting
to each entry of the matrix $\mathbf{P}=\mathbf{P}(\nt,\ud,\la)$.
\end{definition}
\begin{lemma}
\label{lmlf}
Consider an MT~\er{uvf} from~\er{vdde} to~\er{sdde},
the total derivative operator~$\tdt$ corresponding to~\eqref{sdde},
and the total derivative operator~$\tiltdt$ corresponding to~\eqref{vdde}.
Using Definition~\textup{\ref{rlift}}, in this lemma we work with 
liftings by means of the MT~\er{uvf}.

For any scalar functions $h=h(\nt,\ud,\la)$, $\,q=q(\nt,\ud,\la)$, $\,g=g(\nt,\ud,\la)$ 
and $\sm\times\sm$ matrix-valued functions $\mathbf{P}=\mathbf{P}(\nt,\ud,\la)$, 
$\,\mathbf{Q}=\mathbf{Q}(\nt,\ud,\la)$, one has
\begin{gather}
\lb{hqk}
\overline{h+q}=\overline{h}+\overline{q},\qquad\quad
\overline{h\cdot q}=\overline{h}\cdot\overline{q},\qquad\quad
\overline{\cS^k(h)}=\cS^k\big(\overline{h}\big)\qquad\quad\forall\,k\in\zz,\\
\lb{gm1}
\text{if $g\neq 0\,$ then $\,\overline{g}\neq 0\,$ and $\,\overline{g}^{-1}=\overline{g^{-1}}$},\\
\lb{dth}
\overline{\tdt(h)}=\tiltdt\big(\overline{h}\big),\\
\lb{pqs}
\overline{\mathbf{P}+\mathbf{Q}}=
\overline{\mathbf{P}}+\overline{\mathbf{Q}},\qquad\quad
\overline{\mathbf{P}\cdot\mathbf{Q}}=
\overline{\mathbf{P}}\cdot\overline{\mathbf{Q}},\qquad\quad
\overline{\cS^k(\mathbf{P})}=\cS^k\big(\overline{\mathbf{P}}\big)\qquad\quad\forall\,k\in\zz,\\
\lb{tt}
\overline{\tdt(\mathbf{P})}=\tiltdt\big(\overline{\mathbf{P}}\big).
\end{gather}
Also, for any invertible $\sm\times\sm$ matrix $\mathbf{H}=\mathbf{H}(\nt,\ud,\la)$,
the matrix $\overline{\mathbf{H}}$ is invertible as well and\/ 
$\overline{\mathbf{H}}^{-1}=\overline{\mathbf{H}^{-1}}$.
\end{lemma}
\begin{proof}
Property~\er{hqk} follows immediately from Definition~\ref{rlift} of liftings
and the definition of the operator~$\sop$.

To prove~\er{gm1}, consider a nonzero scalar function $g=g(\nt,\ud,\la)$.
Since $\overline{1}=1$ and $1=g\cdot g^{-1}$, using~\er{hqk}, we derive
$1=\overline{1}=\overline{g\cdot g^{-1}}=\overline{g}\cdot\overline{g^{-1}}$,
which yields~\er{gm1}.

Next, the definitions of~$\tdt$,~$\tiltdt$, Definition~\ref{rlift} of liftings, 
and relations~\er{dtmfi} imply that 
for any scalar functions $a=a(\nt,\ud,\la)$, $\,b=b(\nt,\ud,\la)$, 
$\,y=y(\nt,\muv_\ell,\ldots,\la)$, $\,z=z(\nt,\muv_\ell,\ldots,\la)$, $\,c(\nt,\la)$
and any nonzero scalar functions $r=r(\nt,\ud,\la)$, $\,v=v(\nt,\muv_\ell,\ldots,\la)$, 
we have
\begin{gather}
\lb{ab}
\tdt(ab)=\tdt(a)b+a\tdt(b),\,\,\ 
\tiltdt(yz)=\tiltdt(y)z+y\tiltdt(z),\,\,\ 
\tdt\big(c(\nt,\la)\big)=\tiltdt\big(c(\nt,\la)\big)=\pd_t\big(c(\nt,\la)\big),\\
\lb{rvc}
\tdt(r^{-1})=-r^{-2}\tdt(r),\qquad\quad\tiltdt(v^{-1})=-v^{-2}\tiltdt(v),\qquad\quad
\overline{c(\nt,\la)}=c(\nt,\la),\\
\lb{dtu}
\begin{split}
\overline{\tdt(\uv^i_k)}&=\overline{\sop^k\big(\feq^i(\nt,\uv_\una,\uv_{\una+1},\ldots,\uv_\unb)\big)}
=\sop^k\big(\feq^i\big(\nt,\sop^\ua(\mf),\sop^{\ua+1}(\mf),\dots,\sop^\ub(\mf)\big)\big)=\\
&=\sop^k\big(\tiltdt\big(\mf^i(\nt,\muv_\ell,\dots)\big)\big)
=\tiltdt\big(\sop^k\big(\mf^i(\nt,\muv_\ell,\dots)\big)\big)
\qquad\forall\,i=1,\dots,\diunw,\quad\forall\,k\in\zz.
\end{split}
\end{gather}
Note that \er{dtu} gives~\er{dth} in the case $h=\uv^i_k$.

Now, using \er{hqk}, \er{gm1}, \er{ab},  \er{rvc}, \er{dtu}, 
it is straightforward to prove property~\er{dth} 
in the case when $h$ is a polynomial in the 
variables~$n$, $t$, $\uv_\ik=(\uv^1_\ik,\ldots,\uv^{\diu}_\ik)$ with coefficients in~$\mer$.
(To this end, one can use induction on the degree of the polynomial.)

Next, using~\er{dth} for such polynomials, 
it is easy to prove property~\er{dth} in the general case where $h$
is equal to the ratio of two polynomials in~$n$, $t$, $\uv_\ik$ with coefficients in~$\mer$.

Let $\mathrm{Id}$ be the identity matrix of size~$\sm\times\sm$.
Since $\overline{\mathrm{Id}}=\mathrm{Id}$ and 
$\mathrm{Id}=\mathbf{H}\cdot\mathbf{H}^{-1}$, using~\er{pqs}, we get 
\begin{gather}
\notag
\mathrm{Id}=\overline{\mathrm{Id}}=\overline{\mathbf{H}\cdot\mathbf{H}^{-1}}
=\overline{\mathbf{H}}\cdot\overline{\mathbf{H}^{-1}},
\end{gather}
which implies that the matrix $\overline{\mathbf{H}}$ is invertible 
and $\overline{\mathbf{H}}^{-1}=\overline{\mathbf{H}^{-1}}$.
\end{proof}
\begin{theorem}
\lb{thdet2}
For a given equation~\er{sdde},
consider an MLR $\big(\mm(\nt,\ud,\la),\,\lwt(\nt,\ud,\la)\big)$, 
a gauge transformation $\gbf=\gbf(\nt,\ud,\la)$, and the corresponding gauge equivalent MLR
\begin{gather}
\lb{hmwh}
\hmm=\sop(\gbf)\cdot\mm\cdot\gbf^{-1},\qquad\quad
\widehat{\lwt}=\tdt(\gbf)\cdot\gbf^{-1}+\gbf\cdot\lwt\cdot\gbf^{-1}.
\end{gather}
Suppose that one has an MT~\er{uvf} 
from some equation~\er{vdde} to equation~\er{sdde}.
According to Definition~\textup{\ref{rlift}}, we can consider the matrices
\begin{gather}
\lb{lift3}
\overline{\mm}=\overline{\mm}(\nt,\muv_k,\ldots,\la),\qquad\quad 
\overline{\hmm}=\overline{\hmm}(\nt,\muv_k,\ldots,\la),\\
\notag
\overline{\lwt}=\overline{\lwt}(\nt,\muv_k,\ldots,\la),\qquad\quad 
\overline{\widehat{\lwt}}=\overline{\widehat{\lwt}}(\nt,\muv_k,\ldots,\la),\qquad\quad
\overline{\gbf}=\overline{\gbf}(\nt,\muv_k,\ldots,\la)
\end{gather}
which are the liftings of\/ $\mm$, $\,\hmm$, $\,\lwt$, $\,\widehat{\lwt}$, $\,\gbf$
by means of the MT~\er{uvf}.

Then $\big(\overline{\mm},\,\overline{\lwt}\big)$ and 
$\Big(\overline{\hmm},\,\overline{\widehat{\lwt}}\Big)$ are MLRs for equation~\eqref{vdde}.
Furthermore, relations~\er{hmwh} yield
\begin{gather}
\lb{ovhm}
\overline{\hmm}=
\sop\big(\overline{\gbf}\big)\cdot\overline{\mm}\cdot\overline{\gbf}^{-1},\\
\lb{ovwh}
\overline{\widehat{\lwt}}=
\tiltdt\big(\overline{\gbf}\big)\cdot\overline{\gbf}^{-1}
+\overline{\gbf}\cdot\overline{\lwt}\cdot\overline{\gbf}^{-1},
\end{gather}
which implies that $\big(\overline{\mm},\,\overline{\lwt}\big)$ 
is gauge equivalent to $\Big(\overline{\hmm},\,\overline{\widehat{\lwt}}\Big)$.
Here $\tiltdt$ is the total derivative operator corresponding to equation~\eqref{vdde}.

Using formula~\er{foreo} with~$\muv^j$ instead of~$\uv^j$, 
we compute $\eo{\muv^j}\big(\overline{\mm}\big)$ and 
$\eo{\muv^j}\Big(\overline{\hmm}\Big)$.
Then we have 
\begin{gather}
\lb{dettr2}
\begin{gathered}
\det\Big(\eo{\muv^j}\Big(\overline{\hmm}\Big)\Big)=
\det\big(\eo{\muv^j}\big(\overline{\mm}\big)\big),\qquad\quad
\Tr\Big(\Big(\eo{\muv^j}\Big(\overline{\hmm}\Big)\Big)^\pw\Big)=
\Tr\big(\big(\eo{\muv^j}\big(\overline{\mm}\big)\big)^\pw\big),\\
j=1,\ldots,\diunw,\qquad\quad\pw=1,\ldots,\sm-1.
\end{gathered}
\end{gather}
\end{theorem}
\begin{proof}
According to Definition~\ref{dmlpgtnew}, 
the matrices $\mm$, $\gbf$ are invertible, and one has~\er{lrnw}.
Lemma~\ref{lmlf} implies that $\overline{\mm}$ and $\overline{\gbf}$
are invertible as well. 
Applying the operation of lifting to~\er{lrnw}, we get
\begin{gather}
\lb{ovlr}
\overline{\tdt(\mm)}=\overline{\sop(\lwt)\cdot\mm}-\overline{\mm\cdot\lwt},
\end{gather}
By Lemma~\ref{lmlf}, one can rewrite~\er{ovlr} as 
$\tiltdt\big(\overline{\mm}\big)=
\sop\big(\overline{\lwt}\big)\cdot\overline{\mm}-\overline{\mm}\cdot\overline{\lwt}$.
Therefore, $\big(\overline{\mm},\,\overline{\lwt}\big)$ is an MLR for~\eqref{vdde}.

Applying the operation of lifting to~\er{hmwh}, we derive
\begin{gather}
\lb{ovhmw}
\overline{\hmm}=\overline{\sop(\gbf)\cdot\mm\cdot\gbf^{-1}},\qquad\quad
\overline{\widehat{\lwt}}=\overline{\tdt(\gbf)\cdot\gbf^{-1}}
+\overline{\gbf\cdot\lwt\cdot\gbf^{-1}}.
\end{gather}
By Lemma~\ref{lmlf}, relations~\er{ovhmw} can be rewritten as~\er{ovhm},~\er{ovwh}.

Relations~\er{ovhm},~\er{ovwh} say that the pair 
$\Big(\overline{\hmm},\,\overline{\widehat{\lwt}}\Big)$
is obtained from $\big(\overline{\mm},\,\overline{\lwt}\big)$ 
by applying the gauge transformation $\overline{\gbf}$.
Since $\big(\overline{\mm},\,\overline{\lwt}\big)$ is an MLR for equation~\eqref{vdde}, 
we see that $\Big(\overline{\hmm},\,\overline{\widehat{\lwt}}\Big)$ 
is an MLR for~\eqref{vdde} as well and is gauge equivalent 
to the MLR $\big(\overline{\mm},\,\overline{\lwt}\big)$.

Using~\er{ovhm} and 
property~\er{gmg} in the case $M=\overline{\mm}$, $\,G=\overline{\gbf}$ 
with~$\muv^j$ instead of~$\uv^j$, one gets
\begin{gather}
\notag
\eo{\muv^j}\Big(\overline{\hmm}\Big)=
\eo{\muv^j}\Big(\sop\big(\overline{\gbf}\big)\cdot\overline{\mm}\cdot\overline{\gbf}^{-1}\Big)=
\sop\big(\overline{\gbf}\big)\cdot\eo{\muv^j}\big(\overline{\mm}\big)
\cdot\sop\big(\overline{\gbf}\big)^{-1}
\qquad\quad\forall\,j=1,\ldots,\diunw,
\end{gather}
which allows us to apply Lemma~\ref{lab} to $A=\eo{\muv^j}\big(\overline{\mm}\big)$, 
$\,B=\eo{\muv^j}\Big(\overline{\hmm}\Big)$, $\,C=\sop\big(\overline{\gbf}\big)$.
Then \er{dettr} gives~\er{dettr2}.
\end{proof}
\begin{remark}
\lb{rinv}
As shown in Theorem~\ref{thdet},
the functions~\er{inv} are invariant with respect to the action 
of GTs on the set of MLRs $\big(\mm(\nt,\ud,\la),\,\lwt(\nt,\ud,\la)\big)$ 
for a given equation~\er{sdde}. Theorem~\ref{thdet2} implies 
that each MT~\er{uvf} over equation~\er{sdde}
allows one to derive another collection of invariants as follows.

For each MLR $\big(\mm(\nt,\ud,\la),\,\lwt(\nt,\ud,\la)\big)$, we consider
\begin{gather}
\lb{inv2}
\det\big(\eo{\muv^j}\big(\overline{\mm}\big)\big),\qquad
\Tr\big(\big(\eo{\muv^j}\big(\overline{\mm}\big)\big)^\pw\big),\qquad\quad
j=1,\ldots,\diunw,\qquad\pw=1,\ldots,\sm-1,
\end{gather}
where $\overline{\mm}=\overline{\mm}(\nt,\muv_k,\ldots,\la)$
is the lifting of~$\mm(\nt,\ud,\la)$ by means of the MT~\er{uvf}, 
and the computation of~$\eo{\muv^j}\big(\overline{\mm}\big)$ is explained
in Theorem~\ref{thdet2}. The functions~\er{inv2} are invariant 
with respect to the action of GTs $\gbf=\gbf(\nt,\ud,\la)$ in the following sense.

Applying a gauge transformation $\gbf=\gbf(\nt,\ud,\la)$ to 
the MLR $\big(\mm(\nt,\ud,\la),\,\lwt(\nt,\ud,\la)\big)$,
we get the corresponding gauge equivalent MLR
$\big(\hmm(\nt,\ud,\la),\,\widehat{\lwt}(\nt,\ud,\la)\big)$ given by~\er{hmwh}.
Now we compute~\er{inv2} for the matrix~$\hmm$ instead of~$\mm$,
using the lifting~$\overline{\hmm}$ of~$\hmm$ by means of the MT~\er{uvf}.
Then relations~\er{dettr2} show that \er{inv2} are invariant with respect to the action of GTs.
\end{remark}

\begin{remark}
\lb{rnec}
For a given equation~\er{sdde}, 
consider an MLR $\big(\mm(\nt,\ud,\la),\,\lwt(\nt,\ud,\la)\big)$
and \mbox{another} MLR 
$\big(\hmm(\nt,\ud,\la),\,\widehat{\lwt}(\nt,\ud,\la)\big)$.
Theorem~\ref{thdet} implies the following.
Equations~\er{dettr1} are necessary conditions for these two MLRs 
to be gauge equivalent.

Furthermore, each MT~\er{uvf} over equation~\er{sdde}
gives another collection of necessary conditions as follows.
According to Definition~\ref{rlift}, 
we can consider the matrices~\er{lift3} which are the liftings of 
$\mm(\nt,\ud,\la)$, $\,\hmm(\nt,\ud,\la)$ by means of the MT~\er{uvf}.
Theorem~\ref{thdet2} yields the following.
Equations~\er{dettr2} are necessary conditions for the
MLR $\big(\mm(\nt,\ud,\la),\,\lwt(\nt,\ud,\la)\big)$
and the MLR $\big(\hmm(\nt,\ud,\la),\,\widehat{\lwt}(\nt,\ud,\la)\big)$
to be gauge equivalent.
\end{remark}
\begin{theorem}
\lb{tmtriv}
Consider an MLR $\big(\mm(\nt,\ud,\la),\,\lwt(\nt,\ud,\la)\big)$ for equation~\er{sdde}.
Suppose that one has an MT~\er{uvf} from some equation~\er{vdde} to~\er{sdde}.
Consider the matrices
\begin{gather}
\lb{lftmw}
\overline{\mm}=\overline{\mm}(\nt,\muv_k,\ldots,\la),\qquad\quad 
\overline{\lwt}=\overline{\lwt}(\nt,\muv_k,\ldots,\la),
\end{gather}
which are the liftings of\/ $\mm$, $\lwt$ by means of the MT~\er{uvf}.
By Theorem~\textup{\ref{thdet2}}, 
$\big(\overline{\mm},\,\overline{\lwt}\big)$ is a MLR for~\eqref{vdde}.
Using formula~\er{foreo} with~$\muv^j$ instead of~$\uv^j$, 
we compute $\eo{\muv^j}\big(\overline{\mm}\big)$.

Suppose that there is $j\in\{1,\ldots,\diunw\}$ such that 
$\eo{\muv^j}\big(\overline{\mm}\big)\neq 0$.
Then the MLR $(\mm,\lwt)$ is not gauge equivalent to a trivial MLR.
\end{theorem}
\begin{proof}
Suppose that $(\mm,\lwt)$ is gauge equivalent to a trivial MLR. 
According to Definition~\ref{dtrmlp}, this means that there is a 
gauge transformation $\gbf=\gbf(\nt,\ud,\la)$ such that 
the corresponding gauge equivalent MLR~\er{hmwh} is trivial, and hence
\begin{gather}
\lb{dumw}
\dd{\uv^j_\ilk}\big(\hmm\big)=0,\qquad
\dd{\uv^j_\ilk}\big(\widehat\lwt\big)=0\qquad\forall\,j=1,\ldots,\diunw,
\qquad\forall\,\ilk\in\zz.
\end{gather}
By Theorem~\ref{thdet2}, 
the lifted MLR $\Big(\overline{\hmm},\,\overline{\widehat{\lwt}}\Big)$
is gauge equivalent to $\big(\overline{\mm},\,\overline{\lwt}\big)$.
Taking into account the definition of liftings,
property~\er{dumw} implies that $\Big(\overline{\hmm},\,\overline{\widehat{\lwt}}\Big)$ is trivial.

Thus $\big(\overline{\mm},\,\overline{\lwt}\big)$ 
is gauge equivalent to a trivial MLR for equation~\eqref{vdde}.
Then Theorem~\ref{thtriv} with~$\muv^j$ instead of~$\uv^j$ yields
$\eo{\muv^j}\big(\overline{\mm}\big)=0$ for all $j=1,\ldots,\diunw$, 
which contradicts the assumption about $\eo{\muv^j}\big(\overline{\mm}\big)$ in Theorem~\ref{tmtriv}.
\end{proof}
\begin{definition}
\lb{dprem}
Consider an MLR given by matrices $\mm=\mm(\nt,\ud,\la)$ and $\lwt=\lwt(\nt,\ud,\la)$.
We say that \emph{the parameter~$\la$ can be removed from this MLR 
by a gauge transformation}
if there is a gauge transformation $\gbf=\gbf(\nt,\ud,\la)$ such that 
in the corresponding gauge equivalent MLR
\begin{gather}
\lb{hmhu}
\hmm=\sop(\gbf)\cdot \mm\cdot\gbf^{-1},\qquad\quad
\widehat\lwt=\tdt(\gbf)\cdot\gbf^{-1}+\gbf\cdot\lwt\cdot\gbf^{-1}
\end{gather}
the matrices $\hmm$, $\widehat\lwt$ do not depend on~$\la$.

If $\la$ cannot be removed from the MLR 
$\big(\mm(\nt,\ud,\la),\,\lwt(\nt,\ud,\la)\big)$ by any GT, 
then $\la$ is called a \emph{nonremovable} parameter.
\end{definition}

\begin{theorem}
\lb{thpar}
Consider an MLR given by matrices $\mm=\mm(\nt,\ud,\la)$ and $\lwt=\lwt(\nt,\ud,\la)$.
Suppose that at least one of the functions~\er{inv} depends nontrivially on~$\la$.
Then the parameter~$\la$ cannot be removed from this MLR by any GT.
\end{theorem}
\begin{proof}
Suppose that $\la$ can be removed from this MLR by some gauge transformation $\gbf=\gbf(\nt,\ud,\la)$.
According to Definition~\ref{dprem}, this means that 
the matrices $\hmm$, $\widehat\lwt$ given by~\er{hmhu} do not depend on~$\la$.
Then the functions
\begin{gather}
\notag
\det\big(\eo{\uv^j}\big(\hmm\big)\big),\qquad\quad
\Tr\big(\big(\eo{\uv^j}\big(\hmm\big)\big)^\pw\big),\qquad\quad 
j=1,\ldots,\diunw,\qquad\pw=1,\ldots,\sm-1,
\end{gather}
do not depend on~$\la$ either.
Since, by Theorem~\ref{thdet}, we have~\er{dettr1}, 
this contradicts the assumption that at least one of the functions~\er{inv}
depends nontrivially on~$\la$.
\end{proof}

\begin{theorem}
\lb{thpar2}
For a given equation~\er{sdde}, consider an MLR $\big(\mm(\nt,\ud,\la),\,\lwt(\nt,\ud,\la)\big)$.
Suppose that one has an MT~\er{uvf} from some equation~\er{vdde} to~\er{sdde}.
According to Definition~\textup{\ref{defmtt}}, the MT~\er{uvf} does not depend on~$\la$.
According to Definition~\textup{\ref{rlift}}, we can consider the matrix
$\overline{\mm}=\overline{\mm}(\nt,\muv_k,\ldots,\la)$
which is the lifting of\/ 
$\mm(\nt,\ud,\la)$ by means of the MT~\er{uvf}.
Using formula~\er{foreo} with~$\muv^j$ instead of~$\uv^j$, 
we compute $\eo{\muv^j}\big(\overline{\mm}\big)$.

Suppose that at least one of the functions~\er{inv2}
depends nontrivially on~$\la$.
Then the parameter~$\la$ cannot be removed from the MLR 
$\big(\mm(\nt,\ud,\la),\,\lwt(\nt,\ud,\la)\big)$ 
by any gauge transformation $\gbf=\gbf(\nt,\ud,\la)$.
\end{theorem}
\begin{proof}
Suppose that $\la$ can be removed from the MLR 
$\big(\mm(\nt,\ud,\la),\,\lwt(\nt,\ud,\la)\big)$ by some gauge transformation $\gbf=\gbf(\nt,\ud,\la)$.
This means that 
the matrices $\hmm$, $\widehat\lwt$ given by~\er{hmhu} do not depend on~$\la$.

Let $\overline{\hmm}(\nt,\muv_k,\ldots)$ be the lifting of~$\hmm(\nt,\ud)$,
as described in Definition~\ref{rlift}.
Since $\hmm$ and the MT~\er{uvf} do not depend on~$\la$,
the matrix $\overline{\hmm}$ does not depend on~$\la$ either.
This implies that the functions
\begin{gather}
\notag
\det\Big(\eo{\muv^j}\Big(\overline{\hmm}\Big)\Big),\qquad\quad
\Tr\Big(\Big(\eo{\muv^j}\Big(\overline{\hmm}\Big)\Big)^\pw\Big),\qquad\quad 
j=1,\ldots,\diunw,\qquad\pw=1,\ldots,\sm-1,
\end{gather}
do not depend on~$\la$.
Since, by Theorem~\ref{thdet2}, one has~\er{dettr2}, 
this contradicts the assumption that at least one of the functions~\er{inv2}
depends nontrivially on~$\la$.
\end{proof}
\begin{remark}
\lb{ri0mt}
Consider an MLR $\big(\mm(\nt,\ud,\la),\,\lwt(\nt,\ud,\la)\big)$ for equation~\er{sdde}.
Suppose that one has an MT~\er{uvf} from some equation~\er{vdde} to~\er{sdde}.
Consider the lifting $\overline{\mm}=\overline{\mm}(\nt,\muv_k,\ldots,\la)$
of $\mm=\mm(\nt,\ud,\la)$.
Using formula~\er{foreo} with~$\muv^j$ instead of~$\uv^j$, 
we compute $\eo{\muv^j}\big(\overline{\mm}\big)$ and the invariants~\er{inv2}.

Suppose that at least one of the computed invariants~\er{inv2} is nonzero. 
Then we have $\eo{\muv^j}\big(\overline{\mm}\big)\neq 0$ for some~$j$, 
and Theorem~\ref{tmtriv} implies that the considered MLR
is not gauge equivalent to a trivial MLR.
\end{remark}

\subsection{Examples}
\lb{sbex}

\begin{example}
\lb{exy1}
Let $\diunw=1$. Then $\uv=\uv^1$.
Consider the equation
\begin{gather}
\lb{exsh1}
\uv^1_t=\frac{1}{\uv^1_1-\uv^1_{-1}},
\end{gather}
which is a special case of (V4)
in the Yamilov list of integrable Volterra-type differential-difference equations~\cite{yam2006}
and is a symmetry of the (partial difference) quad-graph equation~$\mathrm{H}1$ from 
the Adler--Bobenko--Suris list~\cite{abs2003,MWXTMP2011}.

It is known that equation~\er{exsh1} possesses the following MLR
\begin{gather}
\lb{mlry}
{\mm}=\begin{pmatrix}
 \uv^{1}_{0} & \la-\uv^{1}_{0} \uv^{1}_{1} \\
 1 & -\uv^{1}_{1} \\
\end{pmatrix},
\qquad\quad
{\lwt}=\begin{pmatrix}
 \dfrac{-\uv^{1}_{1}}{\la (\uv^1_1-\uv^1_{-1})} & 
\dfrac{\uv^{1}_{-1} \uv^{1}_{1}}{\la (\uv^{1}_{1}-\uv^{1}_{-1})} \\
 \dfrac{-1}{\la (\uv^1_1-\uv^1_{-1})} & \dfrac{\uv^{1}_{-1}}{\la (\uv^{1}_{1}-\uv^{1}_{-1})} \\
\end{pmatrix}.
\end{gather}
Computing the first invariant in~\er{inv} for the MLR~\er{mlry}, we get
\begin{gather}
\lb{detm}
\det\big(\eo{\uv^1}({\mm})\big)=-\frac{1}{\la^2}(\uv^1_1-\uv^1_{-1})^2.
\end{gather}
Since the function~\er{detm} depends nontrivially on $\lambda$,
Theorem~\ref{thpar} implies that 
the parameter~$\la$ cannot be removed from the MLR~\er{mlry} by any GT.
\end{example}
\begin{example}
\lb{exy2}
Equation~\er{exsh1} possesses also the following MLR
\begin{gather}
\lb{mlry2}
\begin{gathered}
\widetilde{\mm}=\begin{pmatrix}
 \la \uv^{1}_{1}\!-\!\uv^{1}_{0} & 
\la (\uv^{1}_{0})^2\!-\!\la^2\uv^{1}_{0}\uv^{1}_{1}\!+
\!\la (\uv^{1}_{1})^2\!-\!\uv^{1}_{0} \uv^{1}_{1}\!+\!1 \\
 1 & \uv^{1}_{1}\!-\!\la \uv^{1}_{0} \\
\end{pmatrix},\\
\widetilde{\lwt}=\begin{pmatrix}
 \dfrac{\la \uv^{1}_{0}-\uv^{1}_{1}}{\uv^1_1-\uv^1_{-1}} & 
\dfrac{\la \uv^{1}_{-1} \uv^{1}_{0}+\la \uv^{1}_{1} \uv^{1}_{0}%
-(\la\uv^{1}_{0})^2-\uv^{1}_{-1} \uv^{1}_{1}+\la}{\uv^1_1-\uv^1_{-1}} \\
 \dfrac{1}{\uv^{1}_{1}-\uv^{1}_{-1}} & \dfrac{\uv^{1}_{-1}-\la \uv^{1}_{0}}{\uv^{1}_{1}-\uv^{1}_{-1}} \\
\end{pmatrix}.
\end{gathered}
\end{gather}
Computing the first invariant in~\er{inv} for the MLR~\er{mlry2}, we obtain
\begin{gather}
\lb{detbm}
\det\big(\eo{\uv^1}\big(\widetilde{\mm}\big)\big)=-(\uv^1_1-\uv^1_{-1})^2.
\end{gather}
Since~\er{detbm} is not equal to~\er{detm},
Theorem~\ref{thdet}  and Remark~\ref{rnec} imply that 
the MLR~\er{mlry2} is not gauge equivalent to the MLR~\er{mlry}.
\end{example}
\begin{example}
\lb{exy3}
Computing the invariants~\er{inv} for the MLR~\er{mlry2},
one gets functions which do not depend on~$\la$.
This suggests to try to remove the parameter~$\la$ from~\er{mlry2} by some GT.
Indeed, applying to the MLR~\er{mlry2} the gauge transformation $\gbf=
\begin{pmatrix}
 1 & -\la \uv^{1}_{0} \\
 0 & 1 \\
\end{pmatrix}$,
we obtain the gauge equivalent MLR
\begin{gather*}
\widehat\mm=\sop(\gbf)\cdot\widetilde{\mm}\cdot\gbf^{-1}=
\begin{pmatrix}
 -\uv^{1}_{0} & 1-\uv^{1}_{0} \uv^{1}_{1} \\
 1 & \uv^{1}_{1} \\
\end{pmatrix},\qquad
\widehat\lwt=\tdt(\gbf)\cdot\gbf^{-1}+\gbf\cdot\widetilde{\lwt}\cdot\gbf^{-1}=
\begin{pmatrix}
 \dfrac{-\uv^{1}_{1}}{\uv^{1}_{1}-\uv^{1}_{-1}} & \dfrac{-\uv^{1}_{-1} \uv^{1}_{1}}{\uv^{1}_{1}-\uv^{1}_{-1}} \\
 \dfrac{1}{\uv^{1}_{1}-\uv^{1}_{-1}} & \dfrac{\uv^{1}_{-1}}{\uv^{1}_{1}-\uv^{1}_{-1}} \\
\end{pmatrix},
\end{gather*}
which does not depend on~$\la$. 
\end{example}
\begin{example}
\lb{exms}
Consider the non-autonomous equation
\begin{gather}
\lb{eqms}
\uv^1_t=\frac{n}{\uv^1_1-\uv^1_{-1}},
\end{gather}
which is a special case of~\cite[eq.~(28)]{CherdYamil1995}.
It is known that equation~\er{eqms} has the following MLR
\begin{gather}
\lb{mlrms}
{\mm}=\begin{pmatrix}
 \uv^{1}_{0} & \la-t-\uv^{1}_{0} \uv^{1}_{1} \\
 1 & -\uv^{1}_{1} \\
\end{pmatrix},
\qquad\quad
{\lwt}=\begin{pmatrix}
 \dfrac{-n\uv^{1}_{1}}{(\la-t)(\uv^1_1-\uv^1_{-1})} & 
\dfrac{n\uv^{1}_{-1} \uv^{1}_{1}}{(\la-t)(\uv^{1}_{1}-\uv^{1}_{-1})} \\
 \dfrac{-n}{(\la-t)(\uv^1_1-\uv^1_{-1})} & 
\dfrac{n\uv^{1}_{-1}}{(\la-t)(\uv^{1}_{1}-\uv^{1}_{-1})} \\
\end{pmatrix},
\end{gather}
which depends explicitly on~$n,\,t$.
Computing the first invariant in~\er{inv} for the MLR~\er{mlrms}, we get
\begin{gather}
\lb{detms}
\det\big(\eo{\uv^1}({\mm})\big)=-\frac{1}{(\la-t)^2}(\uv^1_1-\uv^1_{-1})^2.
\end{gather}
Since \er{detms} depends nontrivially on~$\lambda$,
Theorem~\ref{thpar} yields that 
the parameter~$\la$ cannot be removed from the MLR~\er{mlrms} by any GT.
\end{example}
\begin{example}
\lb{exgtr}
Since the invariants~\er{detm}, \er{detbm}, \er{detms} 
computed for the MLRs~\er{mlry}, \er{mlry2}, \er{mlrms} are nonzero,
Remark~\ref{ri0} implies that each of these MLRs 
is not gauge equivalent to a trivial MLR.
\end{example}
\begin{example}
\lb{exy4}
Let $\diunw=2$. Then $\uv=(\uv^1,\uv^2)$.
The following $2$-component equation 
is derived in~\cite[Section~5.1.1]{MMBW2013} 
in a study of invariant time evolutions of polygons in the centro-affine plane
\begin{gather}
\label{MMBW}
\left\{
\begin{aligned}
\uv^1_t&=\uv^{1}_{0} \uv^{2}_{0},\\ 
\uv^2_t&=\frac{\uv^{1}_{-1}}{\uv^{1}_{0}}-\frac{\uv^{1}_{0}}{\uv^{1}_{1}}.
\end{aligned}\right.
\end{gather}
The following MLR for equation~\er{MMBW} is obtained in~\cite{IgMLR26}
\begin{gather}
\lb{mmtd}
\mm=\begin{pmatrix}
 {\la}-\uv^{2}_{0} & 1/\uv^{1}_{0}\\
 -\uv^{1}_{0} & 0 \\
\end{pmatrix},\qquad\quad
\lt=\begin{pmatrix}
 0 & -1/\uv^{1}_{0}\\
 \uv^{1}_{-1} & {\la} \\
\end{pmatrix}.
\end{gather}
Fix a constant $\mc\in\mathbb{C}$. The paper~\cite{IgMLR26} presents the equation
\begin{gather}
\label{ceqnew}
\left\{
\begin{aligned}
\muv^1_t&=\muv^{1}_{0} \muv^{2}_{0}/\muv^{2}_{1},\\ 
\muv^2_t&=\muv^{2}_{0} ({\mc} \muv^{1}_{-1}-\muv^{1}_{0})/\muv^{1}_{-1}
\end{aligned}\right.
\end{gather}
and the MT
\begin{gather}
\label{mtmtd}
\left\{
\begin{aligned}
\uv^1&=-\muv^{2}_{1}/\muv^{1}_{0},\\ 
\uv^2&=({\mc}\muv^{1}_{0}\muv^{2}_{1}-\muv^{1}_{0}\muv^{2}_{0}-\muv^{1}_{1}%
\muv^{2}_{1})/(\muv^{1}_{0} \muv^{2}_{1})
\end{aligned}\right.
\end{gather}
from~\er{ceqnew} to~\er{MMBW}.
\begin{lemma}
\label{lnc}
The MT~\er{mtmtd} obeys the nondegeneracy condition from Definition~\textup{\ref{defmtt}}.
\end{lemma}
\begin{proof}
Since here $\diunw=2$, one has $\uv_\ell=(\uv^1_\ell,\uv^2_\ell)$, $\,\ell\in\zz$.
Applying the operator $\sop^\ell$ to~\er{mtmtd}, we get
\begin{gather}
\label{mtdel}
\left\{
\begin{aligned}
\uv^1_\ell&=-\muv^{2}_{\ell+1}/\muv^{1}_{\ell},\\ 
\uv^2_\ell&=({\mc}\muv^{1}_{_\ell}\muv^{2}_{\ell+1}-\muv^{1}_{\ell}\muv^{2}_{\ell}-\muv^{1}_{\ell+1}%
\muv^{2}_{\ell+1})/(\muv^{1}_{\ell} \muv^{2}_{\ell+1}),\qquad\ell\in\zz.
\end{aligned}\right.
\end{gather}
We need to prove that, for any nonzero polynomial $p(\nt,\ud)$,
the corresponding function 
\begin{gather}
\label{ovp}
\text{$\overline{p}=\overline{p}(\nt,\muv_k,\ldots)$ obtained from~$p(\nt,\ud)$ 
by substituting~\er{mtdel}}
\end{gather}
is nonzero as a rational function of 
$n$, $t$ and $\muv_k=(\muv^1_k,\muv^2_k)$, $\,k\in\zz$.

We use induction on the degree of the polynomial $p=p(\nt,\ud)$.
The case $\deg(p)=0$ is obvious.
To make the induction step, we take $m\in\zp$ such that 
the above statement is valid for all nonzero polynomials of degree~$\le m$. 
One needs to prove the statement for every polynomial $p=p(\nt,\ud)$ of degree~$m\!+\!1$. 

First, consider the case when $p$ depends nontrivially on 
$\uv_\ell=(\uv^1_\ell,\uv^2_\ell)$ for some $\ell\in\zz$.
Let $\mai\in\zz$ be the maximal integer such that 
either $\dd{\uv^1_\mai}(p)\neq 0$ or $\dd{\uv^2_\mai}(p)\neq 0$.
Using the induction hypothesis, the form of~\er{mtdel}, 
the definition~\er{ovp} of~$\overline{p}$, and
the chain rule for derivatives, it is easy to show the following:
\begin{itemize}
	\item If $\dd{\uv^2_\mai}(p)\neq 0$, then $\dd{\muv^1_{\mai+1}}(\overline{p})\neq 0$.
	\item If $\dd{\uv^2_\mai}(p)=0$, then $\dd{\uv^1_\mai}(p)\neq 0$ 
	and $\dd{\muv^2_{\mai+1}}(\overline{p})\neq 0$.
\end{itemize}

Finally, if $p$ of degree~$m\!+\!1$ does not depend on 
$\uv_\ell=(\uv^1_\ell,\uv^2_\ell)$ for any $\ell\in\zz$, 
then $\overline{p}$ given by~\er{ovp} is equal to~$p=p(\nt)\neq 0$.
\end{proof}

According to Definition~\ref{rlift}, 
in order to compute the liftings~$\overline{\mm}$, $\overline{\lwt}$ 
of the matrices~$\mm$, $\lwt$ from~\er{mmtd}
by means of the MT~\er{mtmtd}, we need to substitute 
$$
\uv^1_0=-\muv^{2}_{1}/\muv^{1}_{0},\qquad\quad
\uv^2_0=({\mc}\muv^{1}_{0}\muv^{2}_{1}-\muv^{1}_{0}\muv^{2}_{0}-\muv^{1}_{1}%
\muv^{2}_{1})/(\muv^{1}_{0} \muv^{2}_{1}),\qquad\quad
\uv^{1}_{-1}=\sop^{-1}(-\muv^{2}_{1}/\muv^{1}_{0})=-\muv^{2}_{0}/\muv^{1}_{-1}
$$
into the matrices $\mm(\uv^1_0,\uv^2_0,\la)$, $\lt(\uv^1_0,\uv^{1}_{-1},\la)$ given by~\er{mmtd}.
We obtain
\begin{gather}
\lb{ovmm}
\overline{\mm}=\begin{pmatrix}
 {\la}-({\mc}\muv^{1}_{0}\muv^{2}_{1}-\muv^{1}_{0}\muv^{2}_{0}-\muv^{1}_{1}%
\muv^{2}_{1})/(\muv^{1}_{0} \muv^{2}_{1}) 
& -{\muv^{1}_{0}}/{\muv^{2}_{1}}\\
{\muv^{2}_{1}}/{\muv^{1}_{0}} & 0 \\
\end{pmatrix},\\
\lb{ovlwt}
\overline{\lwt}=\begin{pmatrix}
 0 & {\muv^{1}_{0}}/{\muv^{2}_{1}}\\
 -\muv^{2}_{0}/\muv^{1}_{-1} & {\la} \\
\end{pmatrix}.
\end{gather}
In order to use Theorem~\ref{thpar2}, 
we compute the first function in~\er{inv2} for the matrix~\er{ovmm} and get
\begin{gather}
\lb{dtovm}
\det\big(\eo{\muv^1}\big(\overline{\mm}\big)\big)=
(\la-\mc)\muv^{2}_{1}/\big((\muv^{1}_{0})^2\muv^{2}_{0}\big).
\end{gather}
Since the function~\er{dtovm} is nonzero and depends nontrivially on~$\lambda$,
Remark~\ref{ri0mt} implies that the MLR~\er{mmtd}
is not gauge equivalent to a trivial MLR,
and Theorem~\ref{thpar2} yields that 
the parameter~$\la$ cannot be removed from the MLR~\er{mmtd} by any GT.
\end{example}
\begin{example}
\lb{exov}
As discussed in Example~\ref{exy4}, the matrices~\er{mmtd} 
form an MLR for equation~\er{MMBW}, 
while the matrices~\er{ovmm},~\er{ovlwt} are the liftings of~\er{mmtd}
by means of the MT~\er{mtmtd} from~\er{ceqnew} to~\er{MMBW}.
Therefore, by Theorem~\ref{thdet2}, we see that \er{ovmm},~\er{ovlwt}
form an MLR for equation~\er{ceqnew}.

Let us try to apply Theorem~\ref{thpar} to the MLR~\er{ovmm},~\er{ovlwt},
using the dynamical variables $\muv_\ell=(\muv^1_\ell,\muv^2_\ell)$, $\,\ell\in\zz$.
To this end, we compute the first function in~\er{inv} for the matrix~\er{ovmm},
using~$\muv^1$ instead of~$\uv^1$. 
We get~\er{dtovm}, which depends nontrivially on~$\lambda$.
Therefore, by Theorem~\ref{thpar}, the parameter~$\la$ cannot be removed 
from the MLR~\er{ovmm},~\er{ovlwt} by any GT.
Also, since \er{dtovm} is nonzero, 
Remark~\ref{ri0} (considered with~$\muv^j$ instead of~$\uv^j$) 
implies that this MLR is not gauge equivalent to a trivial MLR.
\end{example}
\begin{example}
\lb{extoda}
The Toda lattice in the Flaschka--Manakov coordinates~\cite{FlaschPhR74,Manak75}
can be written as the $2$-component equation
\begin{gather}
\label{todafm}
\left\{
\begin{aligned}
\td^1_t&=\td^{1}_{0} (\td^{2}_{0}-\td^{2}_{1}),\\ 
\td^2_t&=\td^{1}_{-1}-\td^{1}_{0},
\end{aligned}\right.
\end{gather}
where $\td^1=\td^1(n,t)$, $\,\td^2=\td^2(n,t)$.
It is known (see, e.g.,~\cite{kmw}) that the following matrices form an MLR for~\eqref{todafm}
\begin{gather}
\label{mlrtd}
\mm=\begin{pmatrix}
 {\la}-\td^{2}_{1} & \td^{1}_{0} \\
 -1 & 0 \\
\end{pmatrix},
\qquad\quad
\lt=\begin{pmatrix}
 0 & -\td^{1}_{0} \\
 1 & {\la}-\td^{2}_{0} \\
\end{pmatrix}.
\end{gather}

It is shown in~\cite[Section~5.1.1]{MMBW2013} that one has the MT 
\begin{gather}
\label{mtmbw}
\left\{
\begin{aligned}
\td^1&=\uv^{1}_{0}/\uv^{1}_{1},\\ 
\td^2&=\uv^{2}_{0}
\end{aligned}\right.
\end{gather}
from~\eqref{MMBW} to~\er{todafm}.
Recall that we have also the MT~\er{mtmtd} from~\eqref{ceqnew} to~\er{MMBW}.
The composition of the MTs~\er{mtmbw},~\er{mtmtd} gives the following MT 
\begin{gather}
\label{mtcom}
\left\{
\begin{aligned}
\td^1&=\muv^{2}_{1}\muv^{1}_{1}/(\muv^{1}_{0}\muv^{2}_{2}),\\
\td^2&=({\mc}\muv^{1}_{0}\muv^{2}_{1}-\muv^{1}_{0}\muv^{2}_{0}-\muv^{1}_{1}%
\muv^{2}_{1})/(\muv^{1}_{0} \muv^{2}_{1})
\end{aligned}\right.
\end{gather}
from~\eqref{ceqnew} to~\er{todafm}. 
\begin{lemma}
\label{lnc2}
The MT~\er{mtmbw} from~\eqref{MMBW} to~\er{todafm}
and the MT~\er{mtcom} from~\eqref{ceqnew} to~\er{todafm}
obey the nondegeneracy condition from Definition~\textup{\ref{defmtt}}.
\end{lemma}
\begin{proof}
Since here $\diunw=2$, one has $w_\ell=(w^1_\ell,w^2_\ell)$, 
$\,\uv_\ell=(\uv^1_\ell,\uv^2_\ell)$, $\,\muv_\ell=(\muv^1_\ell,\muv^2_\ell)$, $\,\ell\in\zz$.
Applying the operator $\sop^\ell$ to~\er{mtmbw}, we get
\begin{gather}
\label{elmtm}
\left\{
\begin{aligned}
\td^1_\ell&=\uv^{1}_{\ell}/\uv^{1}_{\ell+1},\\ 
\td^2_\ell&=\uv^{2}_{\ell},\qquad\ell\in\zz.
\end{aligned}\right.
\end{gather}
Take a nonzero polynomial $p(\nt,w_\ell,\ldots)$.
Consider the rational function 
\begin{gather}
\label{ovp2}
\text{$\overline{p}=\overline{p}(\nt,\uv_k,\ldots)$ obtained from~$p(\nt,w_\ell,\ldots)$
by substituting~\er{elmtm}}.
\end{gather}
It is easily seen that, 
since $p(\nt,w_\ell,\ldots)$ is nonzero, the form of the substitution~\er{elmtm}
implies that the function~\er{ovp2} is nonzero as well.
Therefore, the MT~\er{mtmbw} from~\eqref{MMBW} to~\er{todafm} satisfies 
the nondegeneracy condition from Definition~\ref{defmtt}.

Now let us prove that the MT~\er{mtcom} from~\eqref{ceqnew} 
to~\er{todafm} obeys the nondegeneracy condition.
Since \er{mtcom} is the composition of the MTs~\er{mtmbw},~\er{mtmtd}, 
while the substitutions~\er{elmtm},~\er{mtdel} are obtained by applying~$\sop^\ell$
to~\er{mtmbw},~\er{mtmtd}, it is sufficient to prove the following:\\
for any nonzero polynomial $p(\nt,w_\ell,\ldots)$, the corresponding function
\begin{gather}
\label{ovp3}
\text{$\overline{\overline{p}}=\overline{\overline{p}}(\nt,\muv_q,\ldots)$ 
obtained from~$\overline{p}(\nt,\uv_k,\ldots)$ in~\er{ovp2}
by substituting~\er{mtdel}}
\end{gather}
is nonzero as a rational function of $n$, $t$ and $\muv_q=(\muv^1_q,\muv^2_q)$, $\,q\in\zz$.

As shown above, the rational function~\er{ovp2} is nonzero, hence we have
\begin{gather}
\lb{rp2}
\overline{p}=\frac{f(\nt,\uv_k,\ldots)}{g(\nt,\uv_k,\ldots)},
\end{gather}
for some nonzero polynomials $f(\nt,\uv_k,\ldots)$, $\,g(\nt,\uv_k,\ldots)$.
By Lemma~\ref{lnc}, the corresponding rational functions
$\overline{f}(\nt,\muv_q,\ldots)$, $\,\overline{g}(\nt,\muv_q,\ldots)$
obtained from $f(\nt,\uv_k,\ldots)$, $\,g(\nt,\uv_k,\ldots)$ 
by substituting~\er{mtdel} are nonzero.
Since, according to~\er{ovp3} and~\er{rp2}, we have 
$\overline{\overline{p}}=\overline{f}(\nt,\muv_q,\ldots)/\overline{g}(\nt,\muv_q,\ldots)$, 
we see that $\overline{\overline{p}}$ is nonzero.
\end{proof}

Now let us compute the lifting~$\overline{\mm}$ of the 
matrix~$\mm(\td^{1}_{0},\td^{2}_{1},\la)$ from~\er{mlrtd} by means of the MT~\er{mtcom}.
By Definition~\ref{rlift}, we need to substitute 
$$
\td^{1}_{0}=\muv^{2}_{1}\muv^{1}_{1}/(\muv^{1}_{0}\muv^{2}_{2}),\qquad
\td^{2}_{1}=\sop\big(({\mc}\muv^{1}_{0}\muv^{2}_{1}-\muv^{1}_{0}\muv^{2}_{0}-\muv^{1}_{1}%
\muv^{2}_{1})/(\muv^{1}_{0} \muv^{2}_{1})\big)=
({\mc}\muv^{1}_{1}\muv^{2}_{2}-\muv^{1}_{1}\muv^{2}_{1}-\muv^{1}_{2}%
\muv^{2}_{2})/(\muv^{1}_{1} \muv^{2}_{2})
$$
into the matrix $\mm=\begin{pmatrix}
 {\la}-\td^{2}_{1} & \td^{1}_{0} \\
 -1 & 0 \\
\end{pmatrix}$.
One gets
\begin{gather}
\lb{ovtdm}
\overline{\mm}=\begin{pmatrix}
 {\la}-({\mc}\muv^{1}_{1}\muv^{2}_{2}-\muv^{1}_{1}\muv^{2}_{1}-\muv^{1}_{2}%
\muv^{2}_{2})/(\muv^{1}_{1} \muv^{2}_{2}) & \muv^{2}_{1}\muv^{1}_{1}/(\muv^{1}_{0}\muv^{2}_{2}) \\
 -1 & 0 \\
\end{pmatrix}.
\end{gather}
To use Theorem~\ref{thpar2}, 
we compute the first function in~\er{inv2} for the matrix~\er{ovtdm} and obtain
\begin{gather}
\lb{dettd}
\det\big(\eo{\muv^1}\big(\overline{\mm}\big)\big)=
(\la-\mc)\muv^{2}_{1}/\big((\muv^{1}_{0})^2\muv^{2}_{0}\big).
\end{gather}
Since \er{dettd} is nonzero and depends nontrivially on $\lambda$,
Remark~\ref{ri0mt} yields that the MLR~\er{mlrtd}
is not gauge equivalent to a trivial MLR,
and Theorem~\ref{thpar2} shows that 
the parameter~$\la$ cannot be removed from the MLR~\er{mlrtd} by any GT.
\begin{remark}
\lb{rmtm}
According to the terminology from Definition~\ref{defmtt},
equations~\er{ceqnew} and~\er{MMBW} can be regarded as Miura-type 
modifications of the ``Toda lattice in the Flaschka--Manakov coordinates'' 
equation~\er{todafm}, since, as discussed above, one has 
the MT~\er{mtcom} from~\eqref{ceqnew} to~\er{todafm}
and the MT~\er{mtmbw} from~\eqref{MMBW} to~\er{todafm}.
Furthermore, \er{ceqnew} can be viewed as a Miura-type modification of~\eqref{MMBW},
since we have the MT~\er{mtmtd} from~\eqref{ceqnew} to~\er{MMBW}.
\end{remark}
\end{example}

\section{A comparison with results of S.Yu.~Sakovich and M.~Marvan 
on zero-curvature representations of PDEs}
\lb{szcr}

In this section we review some results of S.Yu.~Sakovich and 
M.~Marvan on (continuous) zero-curvature representations of partial differential equations (PDEs)
and compare them with our results on semi-discrete MLRs of differential-difference equations.

Let $\diunw\in\zsp$. Below we consider an equation 
for an $\diunw$-component vector-function $u=\big(u^1(x,t),\ldots,u^{\diunw}(x,t)\big)$ 
of two continuous (real or complex) variables~$x,\,t$.
For each $\ixd\in\zp$, we set
\begin{gather}
\lb{uimx}
u^i_{\ixd x}:=\frac{\pd^\ixd u^i}{\pd x^\ixd},\qquad\quad i=1,\dots,\diunw,\qquad
\ixd\in\zp.
\end{gather}
In particular, in the case $\ixd=0$ one has $u^i_{0x}=u^i$ 
and $u_{0x}=(u^1_{0x},\ldots,u^{\diunw}_{0x})=(u^1,\ldots,u^{\diunw})$.

In this section we use the formal theory 
of (1+1)-dimensional evolutionary PDEs, where one regards 
\begin{gather}
\lb{xuu}
x,\qquad\qquad
u_{\ixd x}=(u^1_{\ixd x},\ldots,u^\diunw_{\ixd x}),\quad\ixd\in\zp,
\end{gather}
as independent quantities. 
In what follows, the notation of the type $\fn=\fn(x,\uxd)$ means 
that a function~$\fn$ depends on a finite number of the variables~\er{xuu}.
Such functions are supposed to be analytic. See Remark~\ref{ranf} for more details.

Note that $t$ is not included in~\er{xuu}.
The reason for this is explained in Remark~\ref{rsakt}.

The \emph{total derivative operator $D_x$ with respect to~$x$} 
can be applied to functions of the variables~\er{xuu} by the formula
\begin{gather}
\lb{dx}
D_x\big(\fn(x,\uxd)\big)=\dd{x}(\fn)+
\sum_{\substack{i=1,\ldots,\diunw,\\ \ixd\ge 0}}u^i_{(\ixd+1)x}
\cdot\dd{u^i_{\ixd x}}(\fn).
\end{gather}

Let $\ord\in\zsp$. Consider a (1+1)-dimensional evolutionary PDE
\begin{gather}
\lb{evpde}
u_t=\feqd(x,u_{0x},u_{1x},\ldots,u_{\ord x}),
\end{gather}
where $\feqd$ is an (analytic) $\diunw$-component 
vector-function $\feqd=(\feqd^1,\ldots,\feqd^{\diunw})$, 
which may depend on~$x$ and 
on~$u_{\ixd x}=(u^1_{\ixd x},\ldots,u^\diunw_{\ixd x})$ for $\ixd=0,1,\ldots,\ord$.

Using the components of $u=\big(u^1(x,t),\ldots,u^{\diunw}(x,t)\big)$ 
and of $\feqd=(\feqd^1,\ldots,\feqd^{\diunw})$, one can rewrite~\eqref{evpde} as
\begin{gather}
\notag
u^i_t=\feqd^i(x,u_{0x},\ldots,u_{\ord x}),\qquad\quad
i=1,\ldots,\diunw.
\end{gather}
Using~\er{uimx} and~\er{dx}, we see that for each $\ixd\in\zp$ equation~\er{evpde} implies 
\begin{gather}
\lb{uilpd}
\frac{\pd}{\pd t}(u^i_{\ixd x})=(D_x)^\ixd\big(\feqd^i(x,u_{0x},\ldots,u_{\ord x})\big),
\qquad\quad i=1,\ldots,\diunw,\qquad \ixd\in\zp.
\end{gather}
The \emph{total derivative operator~$D_t$ corresponding to equation~\eqref{evpde}} 
acts on functions of the variables~\er{xuu} as follows
\begin{gather}
\lb{dtpd}
D_t\big(\fn(x,\uxd)\big)=\sum_{\substack{i=1,\ldots,\diunw,\\ \ixd\ge 0}}
(D_x)^\ixd\big(\feqd^i(x,u_{0x},\ldots,u_{\ord x})\big)\cdot
\dd{u^i_{\ixd x}}(\fn).
\end{gather}
Formula~\er{dtpd} reflects the chain rule for the derivative with respect to~$t$, 
taking into account~\er{uilpd}.
It is easily seen (and is well known) that the operators $D_x$ and $D_t$ commute.

\begin{remark}
\label{ranf}
In this section we deal with functions which may depend 
on a finite number of the variables~\eqref{xuu} and a complex parameter~$\la$.
The following convention is used: \\
when we consider a function of the form $\mathbf{h}=\mathbf{h}(x,\uxd,\la)$,
we suppose that there is $\mathrm{k}\in\zp$ such that $\mathbf{h}$ is an analytic function
defined on a domain in the space with the following coordinates
\begin{gather}
\lb{xcr}
x,\qquad\qquad
u_{\ixd x}=(u^1_{\ixd x},\ldots,u^\diunw_{\ixd x}),\quad 
0\le\ixd\le\mathrm{k},\qquad\qquad\la.
\end{gather}
\end{remark}

\begin{definition}
\lb{dzcr}
Let $\sm\in\zsp$. 
Let $\azc=\azc(x,\uxd,\la)$ and $\bzc=\bzc(x,\uxd,\la)$
be $\sm\times\sm$ matrix-valued functions.
Suppose that one has
\begin{gather}
\lb{zcr}
D_x(\bzc)-D_t(\azc)+[\azc,\bzc]=0,
\end{gather}
where $D_x$, $D_t$ are given by~\er{dx},~\er{dtpd}.
Then the pair $(\azc,\bzc)$ is called 
a $\sm\times\sm$ \emph{zero-curvature representation} (ZCR) for equation~\eqref{evpde}.

We say that the matrix $\azc=\azc(x,\uxd,\la)$ is the \emph{$x$-part} of the ZCR $(\azc,\bzc)$.

Now, a \emph{gauge transformation} (GT) is an invertible $\sm\times\sm$ 
matrix-valued function $\gbf=\gbf(x,\uxd,\la)$.
If $(\azc,\bzc)$ is a ZCR for~\eqref{evpde}, then the matrices 
\begin{gather}
\lb{gzcr}
\widehat{\azc}=-D_x(\gbf)\cdot\gbf^{-1}+\gbf\cdot\azc\cdot\gbf^{-1},\qquad\quad
\widehat{\bzc}=-D_t(\gbf)\cdot\gbf^{-1}+\gbf\cdot\bzc\cdot\gbf^{-1}
\end{gather}
form a ZCR for equation~\eqref{evpde} as well.
The ZCR $\big(\widehat{\azc},\,\widehat{\bzc}\big)$ is \emph{gauge equivalent} to the ZCR $(\azc,\bzc)$.
\end{definition}

\begin{remark}
\lb{rcom}
Consider a ZCR $\big(\azc(x,\uxd,\la),\,\bzc(x,\uxd,\la)\big)$, 
a GT $\,\gbf=\gbf(x,\uxd,\la)$, and the corresponding 
gauge equivalent ZCR $\big(\widehat{\azc},\,\widehat{\bzc}\big)$ given by~\er{gzcr}.

Let $\tws{\sm}$ be the vector space of all $\sm\times\sm$ 
matrix-valued functions of the form $\mathrm{Q}=\mathrm{Q}(x,\uxd,\la)$.
The expressions $D_x+\azc$, $\,D_t+\bzc$ as well as $\gbf$ can be regarded as linear operators 
on~$\tws{\sm}$ given by
$$
(D_x+\azc)(\mathrm{Q}):=D_x(\mathrm{Q})+\azc\cdot \mathrm{Q},
\qquad (D_t+\bzc)(\mathrm{Q}):=D_t(\mathrm{Q})+\bzc\cdot \mathrm{Q},\qquad
\gbf(\mathrm{Q}):=\gbf\cdot \mathrm{Q},\qquad \mathrm{Q}\in\tws{\sm}.
$$
Then relation~\er{zcr} says that the operators $D_x+\azc$ and $D_t+\bzc$ commute, i.e.,
$[D_x+\azc,\,D_t+\bzc]=0$. Formulas~\er{gzcr} are equivalent to the equations
$$
D_x+\widehat{\azc}=\gbf(D_x+\azc)\gbf^{-1},\qquad\quad
D_t+\widehat{\bzc}=\gbf(D_t+\bzc)\gbf^{-1},
$$
where each term is viewed as an operator on~$\tws{\sm}$.
\end{remark}

\begin{definition}
\lb{dtrzcr}
A ZCR~$(\azc,\bzc)$ for equation~\er{evpde} is said to be \emph{trivial} 
if $\azc$ does not depend on~$u_{\ixd x}$ for any $\ixd\in\zp$.
Then \er{dtpd} yields $D_t(\azc)=0$, and from~\er{zcr} 
we get $D_x(\bzc)=-[\azc,\bzc]$, which implies that $\bzc$ 
does not depend on~$u_{\ixd x}$ either (as can easily be shown using formula~\er{dx}).
Hence a trivial ZCR does not provide any information about equation~\er{evpde}.
\end{definition}

Following S.Yu.~Sakovich~\cite{sakov95}, for each $j=1,\dots,\diunw$ and 
any $\sm\times\sm$ matrix-valued functions 
$\mathrm{A}=\mathrm{A}(x,\uxd,\la)$, $\,\mathrm{Q}=\mathrm{Q}(x,\uxd,\la)$, we set
\begin{gather}
\lb{oaq}
\tnb_{\mathrm{A}}(\mathrm{Q})=D_x(\mathrm{Q})+[\mathrm{A},\mathrm{Q}],\\
\lb{cea}
\ce{u^j}(\mathrm{A})=\sum_{\ixd\ge 0}(-1)^\ixd\big(\tnb_{\mathrm{A}}\big)^\ixd\big(\pd_{u^j_{\ixd x}}(\mathrm{A})\big),\qquad
\mathrm{A}=\mathrm{A}(x,\uxd,\la),
\qquad j=1,\dots,\diunw.
\end{gather}
Here $\big(\tnb_{\mathrm{A}}\big)^\ixd$ is the $\ixd$th power of the operator $\tnb_{\mathrm{A}}$ defined by~\er{oaq}.

For instance, 
if $\mathrm{A}=\mathrm{A}(x,u^i_{0x},u^i_{1x},\la)$ depends only 
on~$x$, $u^i_{0x}$, $u^i_{1x}$, $\,i=1,\ldots,\diunw$, and on~$\la$, then 
\begin{multline*}
\ce{u^j}\big(\mathrm{A}(x,u^i_{0x},u^i_{1x},\la)\big)=
\sum_{\ixd\ge 0}(-1)^\ixd\big(\tnb_{\mathrm{A}}\big)^\ixd\big(\pd_{u^j_{\ixd x}}(\mathrm{A})\big)=
\big(\tnb_{\mathrm{A}}\big)^0\big(\pd_{u^j_{0x}}(\mathrm{A})\big)
-\big(\tnb_{\mathrm{A}}\big)^1\big(\pd_{u^j_{1x}}(\mathrm{A})\big)=\\
=\pd_{u^j_{0x}}(\mathrm{A})-\Big(D_x\big(\pd_{u^j_{1x}}(\mathrm{A})\big)
+\big[\mathrm{A},\,\pd_{u^j_{1x}}(\mathrm{A})\big]\Big).
\end{multline*}

\begin{remark}
\lb{rsakt}
Following~\cite{sakov95}, we assume that the right-hand side~$\feqd$ 
of~\er{evpde} and matrix-valued functions $\azc$, $\bzc$, $\gbf$ in Definition~\ref{dzcr}
do not depend on~$t$.
\end{remark}

Fix $\sm\in\zsp$. 
In what follows, unless otherwise specified, we consider $\sm\times\sm$ ZCRs for a given equation~\er{evpde}.

Proofs for Propositions~\ref{pinvs},~\ref{ptriv} can be derived from sketches given in~\cite{sakov95}.
(In our opinion, the formulation of Propositions~\ref{pinvs},~\ref{ptriv} is more precise than 
the corresponding statements in~\cite{sakov95}, 
but the main ideas for proofs are presented in~\cite{sakov95}.
The paper~\cite{sakov95} contains also a number of other results, which are not discussed here.)
\begin{proposition}[\cite{sakov95}]
\lb{pinvs}
Consider a $\sm\times\sm$ matrix $\azc=\azc(x,\uxd,\la)$ and an invertible $\sm\times\sm$ 
matrix $\gbf=\gbf(x,\uxd,\la)$.
Set $\widehat{\azc}=-D_x(\gbf)\cdot\gbf^{-1}+\gbf\cdot\azc\cdot\gbf^{-1}$.
Then one has
\begin{gather*}
\ce{u^j}\big(\widehat{\azc}\big)=\gbf\cdot\ce{u^j}(\azc)\cdot\gbf^{-1}\qquad\quad
\forall\,j=1,\ldots,\diunw.
\end{gather*}
\end{proposition}
\begin{proposition}[\cite{sakov95}]
\lb{ptriv}
A ZCR $\big(\azc(x,\uxd,\la),\,\bzc(x,\uxd,\la)\big)$ 
is \textup{(}locally\textup{)} gauge equivalent to a trivial ZCR if and only if\/ 
$\azc=\azc(x,\uxd,\la)$ satisfies $\ce{u^j}(\azc)=0$ for all $j=1,\ldots,\diunw$.

More precisely, the following two statements are valid.

\textup{\textbf{I.}} Suppose that a ZCR $\big(\azc(x,\uxd,\la),\,\bzc(x,\uxd,\la)\big)$ 
is gauge equivalent to a trivial ZCR. Then $\ce{u^j}(\azc)=0$ for all $j=1,\ldots,\diunw$.

\textup{\textbf{II.}} Consider a ZCR $\big(\azc(x,\uxd,\la),\,\bzc(x,\uxd,\la)\big)$.
According to the convention in Remark~\textup{\ref{ranf}}, 
there is $\mathrm{k}\in\zp$ such that $\azc$ and $\bzc$ are analytic functions 
on a domain~$\mathbb{U}$ in the space with coordinates~\er{xcr}. 
Suppose that $\ce{u^j}(\azc)=0$ for all $j=1,\ldots,\diunw$.
Then for each point from the domain~$\mathbb{U}$ 
there is a GT $\,\gbf=\gbf(x,\uxd,\la)$ on an open neighborhood of this point 
such that the corresponding gauge equivalent ZCR 
$\big(\widehat{\azc},\,\widehat{\bzc}\big)$ given by~\er{gzcr} 
is trivial \textup{(}on the open neighborhood\textup{)}.
\end{proposition}

Recall that we have deduced Theorems~\ref{thdet},~\ref{thpar} from Theorem~\ref{th1}.
Similarly, Propositions~\ref{prdet},~\ref{prpar} below follow from Proposition~\ref{pinvs}.
\begin{proposition}[\cite{sakov95}]
\lb{prdet}
For each ZCR $\big(\azc(x,\uxd,\la),\,\bzc(x,\uxd,\la)\big)$, we consider
\begin{gather}
\label{invzc}
\det\big(\ce{u^j}(\azc)\big),\qquad
\Tr\big(\big(\ce{u^j}(\azc)\big)^\pw\big),\qquad\quad
\quad j=1,\ldots,\diunw,\qquad\pw=1,\ldots,\sm-1.
\end{gather}
The functions~\er{invzc} are invariant with respect to the action 
of GTs on the set of ZCRs for a given equation~\er{evpde}, in the following sense. 
Consider a ZCR $\big(\azc(x,\uxd,\la),\,\bzc(x,\uxd,\la)\big)$, 
a gauge transformation $\gbf=\gbf(x,\uxd,\la)$, and the corresponding gauge equivalent ZCR 
$\big(\widehat{\azc}(x,\uxd,\la),\,\widehat{\bzc}(x,\uxd,\la)\big)$ given by~\er{gzcr}.
Then one has 
\begin{gather*}
\det\big(\ce{u^j}\big(\widehat{\azc}\big)\big)=\det\big(\ce{u^j}(\azc)\big),\qquad\quad
\Tr\big(\big(\ce{u^j}\big(\widehat{\azc}\big)\big)^\pw\big)=\Tr\big(\big(\ce{u^j}(\azc)\big)^\pw\big),\\
j=1,\ldots,\diunw,\qquad\quad\pw=1,\ldots,\sm-1.
\end{gather*}
\end{proposition}
\begin{proposition}[\cite{sakov95}]
\lb{prpar}
Consider a ZCR given by matrices $\azc=\azc(x,\uxd,\la)$ and $\bzc=\bzc(x,\uxd,\la)$.
Suppose that at least one of the functions~\er{invzc} depends nontrivially on~$\la$.
Then the parameter~$\la$ cannot be removed from this ZCR by any GT.
\end{proposition}

Formulas~\er{nbmq},~\er{foreo},
invariants~\er{inv}, and Theorems~\ref{th1},~\ref{thtriv},~\ref{thdet},~\ref{thpar} 
in Section~\ref{ssdmlr} 
can be regarded as ``semi-discrete (differential-difference) analogs''
of Sakovich's formulas~\er{oaq},~\er{cea}, invariants~\er{invzc}
and Sakovich's results in Propositions~\ref{pinvs},~\ref{ptriv},~\ref{prdet},~\ref{prpar},
which consider the continuous (partial differential) case.
Note that
\begin{itemize}
	\item the operator $\nb_{M}$ defined by~\er{nbmq} is invertible,
while the operator $\tnb_{A}$ given by~\er{oaq} is not invertible;
\item the structure of~\er{foreo} is quite different from that of~\er{cea}.
\end{itemize}

Furthermore, in Section~\ref{ssdmlr} we have also presented another collection of results:\\
as explained in Theorems~\ref{thdet2},~\ref{thpar2},
Remarks~\ref{rinv},~\ref{rnec}, and Examples~\ref{exy4},~\ref{extoda},
we use MTs, in order to derive extra invariants for MLRs 
and to obtain extra results on nonremovability of parameters in MLRs 
in the semi-discrete case.
S.Yu.~Sakovich~\cite{sakov95} does not present such applications of MTs 
in the continuous (partial differential) case.

S.Yu.~Sakovich's results from~\cite{sakov95} described here 
consider ZCRs for a (1+1)-dimensional evolutionary PDE~\er{evpde}.
Independently of Sakovich,
M.~Marvan~\cite{marvan93,marvan97,MarvCohSpPar02,MarvSpPar10} 
developed a theory (using jet spaces and cohomology associated with ZCRs) 
which gives similar results (as well as generalizations and extensions) 
on ZCRs for a much more general class of PDEs (including non-evolutionary PDEs).
To our knowledge, M.~Marvan's works do not consider the above-mentioned applications of MTs.


\begin{remark}
\lb{rcycl}
Note that, in addition to the results reviewed in this section,
S.Yu.~Sakovich's works (\cite{sakov95,sakov2004,SakTrLax} and references therein) 
contain also many more developments of 
gauge-invariant approaches to ZCRs for (1+1)-dimensional evolutionary PDEs, 
including a theory of cyclic bases for ZCRs, which has applications in the following directions:
\begin{itemize}
\item constructing PDEs having ZCRs with a given $x$-part;
\item obtaining hierarchies and recursion operators for given spectral problems;
\item distinguishing between ``true'' and ``fake'' ZCRs (Lax pairs) for PDEs.
\end{itemize}
In order to define a cyclic basis for a given ZCR $\big(\azc(x,\uxd,\la),\,\bzc(x,\uxd,\la)\big)$,
Sakovich uses the matrices 
\begin{gather}
\label{cbzcr}
\big(\tnb_{\azc}\big)^{\rpw}\big(\ce{u^j}(\azc)\big),\qquad\quad
j=1,\ldots,\diunw,\qquad\rpw\in\zp,
\end{gather}
defined by means of~\er{oaq},~\er{cea}. 
The matrix $\big(\tnb_{\azc}\big)^0\big(\ce{u^j}(\azc)\big)=\ce{u^j}(\azc)$ 
is included in~\er{cbzcr}. See~\cite{sakov95,sakov2004,SakTrLax} for details.

Using constructions and results from Section~\ref{ssdmlr},
it is possible to develop an analogous theory of cyclic bases for 
semi-discrete MLRs of differential-difference equations. 
To this end, instead of~\er{cbzcr}, 
for a given MLR $\big(\mm(\nt,\ud,\la),\,\lwt(\nt,\ud,\la)\big)$ one can use the matrices
\begin{gather*}
\nb_{\mm}^\ki\big(\eo{\uv^j}(\mm)\big),\qquad\quad j=1,\ldots,\diunw,\qquad\ki\in\zz,
\end{gather*}
defined by means of~\er{nbmq},~\er{foreo}. This will be described elsewhere.
\end{remark}

\section{Lax--Miura type constructions for partial difference equations}
\lb{sdiscr}

In this section, we consider functions of two integer variables~$n,m$.
For any such function $w=w(n,m)$ and any fixed integers~$i,j$,
we denote by $w_{i,j}$ the function defined as follows $w_{i,j}(n,m):=w(n+i,m+j)$.
In particular, $w_{0,0}=w$.

Furthermore, we denote by~$\sop$ and~$\tso$ the shift operators with respect 
to~$n$ and~$m$, respectively. One has 
\begin{gather}
\notag
\sop(n)=n\!+\!1,\quad\sop(m)=m,\quad\tso(n)=n,\quad\tso(m)=m\!+\!1,\quad
\sop(w_{i,j})=w_{i+1,j},\quad\tso(w_{i,j})=w_{i,j+1}.
\end{gather}
For any $k,\ell\in\zz$, we have the corresponding powers $\sop^k$,~$\tso^\ell$ 
of these operators and the relation $\sop^k\tso^\ell(w_{i,j})=w_{i+k,j+\ell}$.

Let $\diunw,L\in\zsp$. 
Below we consider a partial difference (discrete) equation 
for an $\diunw$-component vector-function $\du=\du(n,m)$.
Then $\du_{i,j}$, $\,i,j\in\zz$, can be called dynamical variables.
(According to our notation, $\du_{i,j}=\sop^i\tso^j(\du)$.)

When we take a function of the form $\fn=\fn(n,m,\du_{k,\ell},\ldots,\la)$,
we assume that $\fn$ may depend on~$n,m$, 
on a finite number of the dynamical variables $\du_{k,\ell}$, $\,k,\ell\in\zz$, 
and a complex parameter~$\la$.

Consider a partial difference equation 
\begin{gather}
\lb{diseq}
\fde(n,m,\du_{i,j},\ldots)=0,
\end{gather}
where $\fde(n,m,\du_{i,j},\ldots)$ is an $L$-component vector-function.
(Thus \er{diseq} is an $L$-component equation.)
Any solution $\du=\du(n,m)$ of~\er{diseq} satisfies also the following equations
\begin{gather}
\lb{sheq}
\sop^k\tso^\ell\big(\fde(n,m,\du_{i,j},\ldots)\big)=0,\qquad\quad k,\ell\in\zz,
\end{gather}
which are called the \emph{shifts} of~\er{diseq}.

The main purpose of this section is to present concrete 
examples (Examples~\ref{exdgt}--\ref{exdlp} below).
Definition~\ref{ddmlr} of DLPs 
is applied to the DLPs in Examples~\ref{exdgt}--\ref{exdlp}.
For them it is straightforward to verify the conditions of Definition~\ref{ddmlr}.



\begin{definition}
\lb{ddmlr}
Let $\sm\in\zsp$. A \emph{discrete $\sm\times\sm$ matrix Lax pair} (DLP) 
for equation~\er{diseq} is given by $\sm\times\sm$ matrices 
$\md=\md(n,m,\du_{k,\ell},\ldots,\la)$ and $\kd=\kd(n,m,\du_{k,\ell},\ldots,\la)$ 
satisfying the following conditions.
\begin{itemize}
	\item The (discrete) Lax equation
\begin{gather}
\notag
\tso(\md)\cdot\kd=\sop(\kd)\cdot\md
\end{gather}
holds modulo equation~\er{diseq} and its shifts~\er{sheq}.
\item The matrices $\md$, $\kd$ are invertible in the following sense: \\
there are matrices 
$\widetilde{\md}=\widetilde{\md}(n,m,\du_{k,\ell},\ldots,\la)$, 
$\,\widetilde{\kd}=\widetilde{\kd}(n,m,\du_{k,\ell},\ldots,\la)$ such that 
the relations
\begin{gather}
\notag
\md\cdot\widetilde{\md}=\widetilde{\md}\cdot\md=\mathrm{Id},\qquad\quad
\kd\cdot\widetilde{\kd}=\widetilde{\kd}\cdot\kd=\mathrm{Id}
\end{gather}
hold modulo~\er{diseq},~\er{sheq}.
Here $\mathrm{Id}$ is the identity matrix of size~$\sm\times\sm$.
\end{itemize}

These conditions imply that the (overdetermined) auxiliary linear system
\begin{gather}
\notag
\sop(\Psi)=\md\cdot\Psi,\qquad\quad
\tso(\Psi)=\kd\cdot\Psi
\end{gather}
is compatible modulo equation~\eqref{diseq}. 
Here $\Psi=\Psi(n,m)$ is a $\sm\times\sm$ matrix-valued function.
\end{definition}

In what follows, we assume $\diunw=L=1$.
Using the function $\fde=\fde(n,m,\du_{i,j},\ldots)$ from~\er{diseq},
we consider the finite set $\mathbb{P}(\fde):=\{(i,j)\in\mathbb{Z}^2\ |\ \dd{\du_{i,j}}(\fde)\neq 0\}$ 
and let $q:=|\mathbb{P}(\fde)|\in\zp$. 
Then \er{diseq} is called a \emph{$q$-point equation}.
For instance:
\begin{itemize}
\item \er{eqly} is a 4-point equation, 
since the left-hand side of~\er{eqly} depends on $4$~dynamical 
variables $\du_{0,0}$, $\du_{1,0}$, $\du_{0,1}$, $\du_{1,1}$.
\item equation~\er{kmxp} for a scalar function $\pp=\pp(n,m)$
is a 6-point equation, 
since the left-hand side of~\er{kmxp} depends on $6$~dynamical 
variables $\pp_{-1,1}$, $\pp_{0,0}$, $\pp_{1,0}$, $\pp_{0,1}$, $\pp_{1,1}$, $\pp_{2,0}$.
\end{itemize}

\begin{remark}
\lb{rreg}
In Examples~\ref{exdgt}--\ref{exdlp}, 
when we make computations with some rational expressions, DLPs, and GTs,
we assume (as a regularity assumption) that the computations are performed 
in the domain where 
\begin{itemize}
	\item the denominators of the considered fractions are nonzero,
	\item the determinants of the matrices in the considered DLPs and GTs are nonzero.
\end{itemize}
\end{remark}

\begin{example}
\lb{exdgt}
Let $\aly(m)$ be a rational scalar function.
D.~Levi and R.I.~Yamilov~\cite[Section~3]{LeYaJPA2009} present 
the following equation 
\begin{gather}
\lb{eqly}
(\du_{1,0}+\aly(m))(\du_{0,0}-\aly(m))
-(\du_{1,1}-\aly(m+1))(\du_{0,1}+\aly(m+1))=0
\end{gather}
with the DLP
\begin{gather}
\lb{dlply}
\begin{gathered}
\md(m,\du_{0,0},\du_{1,0},\la)=\begin{pmatrix}
\la & -(\du_{1,0}\!+\!\aly(m))(\du_{0,0}\!-\!\aly(m))  \\
 1 & 0
\end{pmatrix},\\
\kd=\begin{pmatrix}
\la & 2 \aly(m\!+\!1) (\du_{0,1}\!+\!\aly(m\!+\!1)) \\
 \dfrac{2 \aly(m+1)}{\aly(m+1)-\du_{0,1}} & \dfrac{\la (\du_{0,1}+\aly(m+1))}{\du_{0,1}-\aly(m+1)}
\end{pmatrix}.
\end{gathered}
\end{gather}
\begin{remark}
\lb{rlam}
In~\cite[Section~3]{LeYaJPA2009} this DLP is written differently, using 
the expression $(\lambda-\lambda^{-1})$ for the parameter.
In~\er{dlply} we write one symbol~$\la$ instead of 
the expression $(\lambda-\lambda^{-1})$ used in~\cite[Section~3]{LeYaJPA2009}. 
It is easy to verify that \er{dlply} is a DLP for~\er{eqly}.
\end{remark}

Using the fact that \er{dlply} is a DLP for~\er{eqly},
it is easily shown that, 
for any invertible $2\times 2$ matrix of the form $\gbf=\gbf(m,\du_{0,0},\la)$,
the matrices 
\begin{gather}
\lb{hmd}
\hmd=\sop\big(\gbf(m,\du_{0,0},\la)\big)\cdot\md(m,\du_{0,0},\du_{1,0},\la)
\cdot\big(\gbf(m,\du_{0,0},\la)\big)^{-1},\\
\lb{hkd}
\widehat{\kd}=\tso\big(\gbf(m,\du_{0,0},\la)\big)\cdot\kd
\cdot\big(\gbf(m,\du_{0,0},\la)\big)^{-1}
\end{gather}
constitute a DLP for equation~\er{eqly} as well.
Using the terminology from the literature on integrable equations,
we say that the DLP~\er{hmd},~\er{hkd} 
is obtained from the DLP~\er{dlply}
by applying a gauge transformation (GT) ~$\,\gbf$.
Also, one can say that the DLP~\er{hmd},~\er{hkd} is gauge equivalent to~\er{dlply}.

Let us try to simplify the DLP~\er{dlply} by applying a suitable GT of this type.
Namely, let us try to find a GT $\,\gbf=\gbf(m,\du_{0,0},\la)$ such that 
the matrix~$\hmd$ given by~\er{hmd} depends only on~$m,\,\du_{0,0},\,\la$.
That is, we want to find an invertible $2\times 2$ matrix~$\gbf(m,\du_{0,0},\la)$ 
such that the right-hand side of~\er{hmd} does not depend on~$\du_{1,0}$. 
Thus we require
\begin{gather}
\label{pdu11}
\dd{\du_{1,0}}
\Big(\sop\big(\gbf(m,\du_{0,0},\la)\big)\cdot\md(m,\du_{0,0},\du_{1,0},\la)
\cdot\gbf(m,\du_{0,0},\la)^{-1}\Big)=0.
\end{gather}
Since $\gbf(m,\du_{0,0},\la)^{-1}$ does not depend on~$\du_{1,0}$,
equation~\er{pdu11} is equivalent to
\begin{gather}
\notag
\dd{\du_{1,0}}
\Big(\sop\big(\gbf(m,\du_{0,0},\la)\big)\cdot\md(m,\du_{0,0},\du_{1,0},\la)\Big)=0,
\end{gather}
which says that
\begin{gather}
\label{pdu11p}
\dd{\du_{1,0}}
\Big(\sop\big(\gbf(m,\du_{0,0},\la)\big)\Big)\cdot\md(m,\du_{0,0},\du_{1,0},\la)+
\sop\big(\gbf(m,\du_{0,0},\la)\big)\cdot\dd{\du_{1,0}}
\big(\md(m,\du_{0,0},\du_{1,0},\la)\big)=0.
\end{gather}
Multiplying~\er{pdu11p} by~$\md(m,\du_{0,0},\du_{1,0},\la)^{-1}$ from the right, 
we can rewrite~\er{pdu11p} as
\begin{gather}
\label{pdu11e}
\dd{\du_{1,0}}\Big(\sop\big(\gbf(m,\du_{0,0},\la)\big)\Big)=
-\sop\big(\gbf(m,\du_{0,0},\la)\big)\cdot\dd{\du_{1,0}}\big(\md(m,\du_{0,0},\du_{1,0},\la)\big)
\cdot\md(m,\du_{0,0},\du_{1,0},\la)^{-1}.
\end{gather}
Applying the operator~$\sop^{-1}$ to~\er{pdu11e} and using the relation 
$\sop^{-1}\Big(\dd{\du_{1,0}}\Big(\sop\big(\gbf(m,\du_{0,0},\la)\big)\Big)\Big)
=\dd{\du_{0,0}}\big(\gbf(m,\du_{0,0},\la)\big)$, we get the differential equation
\begin{gather}
\label{pdu11d}
\dd{\du_{0,0}}\big(\gbf(m,\du_{0,0},\la)\big)=
-\gbf(m,\du_{0,0},\la)\cdot\sop^{-1}\Big(\dd{\du_{1,0}}\big(\md(m,\du_{0,0},\du_{1,0},\la)\big)
\cdot\md(m,\du_{0,0},\du_{1,0},\la)^{-1}\Big)
\end{gather}
for an unknown matrix~$\gbf(m,\du_{0,0},\la)$.
For the matrix $\md(m,\du_{0,0},\du_{1,0},\la)$ from~\er{dlply} one has
\begin{gather}
\label{pdu11mu}
\sop^{-1}\Big(\dd{\du_{1,0}}\big(\md(m,\du_{0,0},\du_{1,0},\la)\big)
\cdot\md(m,\du_{0,0},\du_{1,0},\la)^{-1}\Big)=
\begin{pmatrix}
\dfrac{1}{\du_{0,0}+\aly(m)} & -\dfrac{\la}{\du_{0,0}+\aly(m)} \\
 0 & 0
\end{pmatrix}.
\end{gather}
The elements of the $2\times 2$ matrix~$\gbf(m,\du_{0,0},\la)$ 
are denoted by~$g^{ij}=g^{ij}(m,\du_{0,0},\la)$, $\,i,j=1,2$.
Substituting~\er{pdu11mu} and 
$\gbf(m,\du_{0,0},\la)=\begin{pmatrix}
g^{11}(m,\du_{0,0},\la)&g^{12}(m,\du_{0,0},\la)\\
g^{21}(m,\du_{0,0},\la)&g^{22}(m,\du_{0,0},\la)
\end{pmatrix}$
into equation~\er{pdu11d}, we get
\begin{gather}
\label{pdu11d4}
\dd{\du_{0,0}}\begin{pmatrix}
g^{11}(m,\du_{0,0},\la)&g^{12}(m,\du_{0,0},\la)\\
g^{21}(m,\du_{0,0},\la)&g^{22}(m,\du_{0,0},\la)
\end{pmatrix}=
-\begin{pmatrix}
g^{11}&g^{12}\\
g^{21}&g^{22}
\end{pmatrix}
\begin{pmatrix}
\dfrac{1}{\du_{0,0}+\aly(m)} & -\dfrac{\la}{\du_{0,0}+\aly(m)} \\
 0 & 0
\end{pmatrix}.
\end{gather}

It is sufficient to find one (preferably as simple as possible) invertible matrix
$$
\gbf(m,\du_{0,0},\la)=\begin{pmatrix}
g^{11}(m,\du_{0,0},\la)&g^{12}(m,\du_{0,0},\la)\\
g^{21}(m,\du_{0,0},\la)&g^{22}(m,\du_{0,0},\la)
\end{pmatrix}
$$
satisfying~\er{pdu11d4}. Equation~\er{pdu11d4} allows us to take
\begin{gather}
\label{g2122}
g^{21}(m,\du_{0,0},\la)=0,\qquad\qquad g^{22}(m,\du_{0,0},\la)=1.
\end{gather}
Taking into account~\er{g2122}, we can rewrite~\er{pdu11d4} as the system
\begin{gather}
\label{pdu11g11}
\frac{\pd g^{11}}{\pd \du_{0,0}}=
\frac{-g^{11}}{\du_{0,0}+\aly(m)},\\
\label{pdu11g12}
\frac{\pd g^{12}}{\pd \du_{0,0}}=
\frac{\la g^{11}}{\du_{0,0}+\aly(m)},
\end{gather}
where $g^{11}=g^{11}(m,\du_{0,0},\la)$, $\,g^{12}=g^{12}(m,\du_{0,0},\la)$.
It is sufficient to find some functions $g^{11}$, $\,g^{12}$ such that
$g^{11}$, $\,g^{12}$ satisfy~\er{pdu11g11},~\er{pdu11g12}
and the function~$g^{11}$ is nonzero 
(so that $\gbf(m,\du_{0,0},\la)=\begin{pmatrix}
g^{11}&g^{12}\\
0&1
\end{pmatrix}$ is invertible).

From~\er{pdu11g11} it is clear that one can take
\begin{gather}
\label{g11uu}
g^{11}(m,\du_{0,0},\la)=\frac{1}{\du_{0,0}+\aly(m)}.
\end{gather}
Substituting~\er{g11uu} into equation~\er{pdu11g12}, we obtain
\begin{gather}
\label{pdu10g12}
\frac{\pd g^{12}}{\pd\du_{0,0}}=
\frac{\la}{(\du_{0,0}+\aly(m))^2}.
\end{gather}
Integrating~\er{pdu10g12}, we see that one can take
$g^{12}(m,\du_{0,0},\la)=-\dfrac{\la}{\du_{0,0}+\aly(m)}$.
Thus we have constructed the following matrix
\begin{gather}
\label{mgu0}
\gbf(m,\du_{0,0},\la)=\begin{pmatrix}
\dfrac{1}{\du_{0,0}+\aly(m)} & -\dfrac{\la}{\du_{0,0}+\aly(m)} \\
 0 & 1
\end{pmatrix},
\end{gather}
which can be used as a gauge transformation.

Applying the gauge transformation~\er{mgu0} to the DLP~\er{dlply}
as written in~\er{hmd},~\er{hkd}, we get 
\begin{gather}
\notag
\hmd=\sop\big(\gbf(m,\du_{0,0},\la)\big)\cdot\md(m,\du_{0,0},\du_{1,0},\la)
\cdot\big(\gbf(m,\du_{0,0},\la)\big)^{-1}=
\begin{pmatrix}
 0 & \aly(m)-\du_{0,0} \\
 \du_{0,0}+\aly(m) & \la
\end{pmatrix},\\
\notag
\widehat{\kd}=\tso\big(\gbf(m,\du_{0,0},\la)\big)\cdot\kd
\cdot\big(\gbf(m,\du_{0,0},\la)\big)^{-1}
=\begin{pmatrix}
\dfrac{\la (\du_{0,0}+\aly(m))}{\du_{0,1}-\aly(m+1)} & 2 \aly(m+1) \\
 \dfrac{2 \aly(m+1) (\du_{0,0}+\aly(m))}{\aly(m+1)-\du_{0,1}} & \la
\end{pmatrix}.
\end{gather}
The obtained DLP
\begin{gather}
\lb{newdlp}
\hmd=\begin{pmatrix}
 0 & \aly(m)-\du_{0,0} \\
 \du_{0,0}+\aly(m) & \la
\end{pmatrix},
\qquad\quad
\widehat{\kd}=\begin{pmatrix}
\dfrac{\la (\du_{0,0}+\aly(m))}{\du_{0,1}-\aly(m+1)} & 2 \aly(m+1) \\
 \dfrac{2 \aly(m+1) (\du_{0,0}+\aly(m))}{\aly(m+1)-\du_{0,1}} & \la
\end{pmatrix}
\end{gather}
is simpler than the DLP~\er{dlply}.
The fact that the matrix $\hmd$ in~\er{newdlp} depends only on~$m,\,\du_{0,0},\,\la$ 
will allow us to construct a Miura-type transformation (MT) for equation~\er{eqly}, 
as shown in Example~\ref{exdmt}.
\end{example}
\begin{example}
\lb{exdmt}
Recall that \er{newdlp} is a DLP for equation~\eqref{eqly}.
Computing the determinants of the matrices~\er{newdlp}, we get
\begin{gather}
\lb{dethmk}
\det\big(\hmd\big)=(\du_{0,0})^2-\aly(m)^2,
\qquad\quad
\det\big(\widehat{\kd}\big)=\frac{(\la^2+4\aly(m+1)^2)(\du_{0,0}+\aly(m))}%
{\du_{0,1}-\aly(m+1)}.
\end{gather}

Consider a constant $\mc\in\mathbb{C}$.
Denote by $\hmd^\mc$, $\widehat{\kd}^\mc$ 
the matrices~\er{newdlp} with~$\la$ replaced by~$\mc$.
We have 
\begin{gather}
\lb{mcdlp}
\hmd^\mc=\begin{pmatrix}
 0 & \aly(m)-\du_{0,0} \\
 \du_{0,0}+\aly(m) & \mc
\end{pmatrix},
\qquad\quad
\widehat{\kd}^\mc=\begin{pmatrix}
\dfrac{\mc (\du_{0,0}+\aly(m))}{\du_{0,1}-\aly(m+1)} & 2 \aly(m+1) \\
 \dfrac{2 \aly(m+1) (\du_{0,0}+\aly(m))}{\aly(m+1)-\du_{0,1}} & \mc
\end{pmatrix},\\
\lb{detmc}
\det\big(\hmd^\mc\big)=(\du_{0,0})^2-\aly(m)^2,
\qquad\quad
\det\big(\widehat{\kd}^\mc\big)=\frac{(\mc^2+4\aly(m+1)^2)(\du_{0,0}+\aly(m))}%
{\du_{0,1}-\aly(m+1)}.
\end{gather}
Since \er{newdlp} is a DLP for~\eqref{eqly},
the matrices~\er{mcdlp} form a DLP for equation~\eqref{eqly} as well 
(a DLP without~$\la$), provided that the determinants~\er{detmc}
are nonzero as rational functions of~$m,\,\du_{0,0},\,\du_{0,1}$.
To this end, we choose and fix a constant~$\mc\in\mathbb{C}$ as follows:
\begin{itemize}
	\item If the function $\aly=\aly(m)$ is constant, 
	then we take any~$\mc$ satisfying $\mc^2+4\aly^2\neq 0$.
	\item If the rational function $\aly(m)$ is nonconstant, any~$\mc$ is suitable.
\end{itemize}
Then \er{mcdlp} is a DLP for~\eqref{eqly}. 
According to Definition~\ref{ddmlr}, we can consider the auxiliary system
\begin{gather}
\label{syspsit}
\cS(\Psi)=\begin{pmatrix}
 0 & \aly(m)-\du_{0,0} \\
 \du_{0,0}+\aly(m) & \mc
\end{pmatrix}\cdot\Psi,\qquad
\tso(\Psi)=\begin{pmatrix}
\dfrac{\mc (\du_{0,0}+\aly(m))}{\du_{0,1}-\aly(m+1)} & 2 \aly(m+1) \\
 \dfrac{2 \aly(m+1) (\du_{0,0}+\aly(m))}{\aly(m+1)-\du_{0,1}} & \mc
\end{pmatrix}\cdot\Psi,
\end{gather}
which is compatible modulo equation~\eqref{eqly}. 
Here $\Psi=\Psi(n,m)$ is a $2\times 2$ matrix
whose elements we denote by $\uppsi^{ij}=\uppsi^{ij}(n,m)$, $\,i,j=1,2$.

Substituting $\Psi=\begin{pmatrix}
\uppsi^{11}&\uppsi^{12}\\
\uppsi^{21}&\uppsi^{22}
\end{pmatrix}$ into~\eqref{syspsit} and considering the first columns 
in the matrix equations~\eqref{syspsit}, one gets
\begin{gather}
\label{sp11}
\cS(\uppsi^{11})=(\aly(m)-\du_{0,0})\uppsi^{21},\qquad\quad
\cS(\uppsi^{21})=(\du_{0,0}+\aly(m))\uppsi^{11}+\mc\uppsi^{21},\\
\label{dp11}
\tso(\uppsi^{11})=\frac{\mc (\du_{0,0}+\aly(m))}{\du_{0,1}-\aly(m+1)}\uppsi^{11}+
2 \aly(m+1)\uppsi^{21},\qquad\quad
\tso(\uppsi^{21})=\frac{2 \aly(m+1) (\du_{0,0}+\aly(m))}{\aly(m+1)-\du_{0,1}}\uppsi^{11}+
\mc\uppsi^{21}.
\end{gather}
Set $\mdu:={\uppsi^{11}}/{\uppsi^{21}}$.
Using~\eqref{sp11},~\eqref{dp11} and $\mdu_{0,0}=\mdu={\uppsi^{11}}/{\uppsi^{21}}$, we derive
\begin{gather}
\label{swp11}
\cS(\mdu)=\frac{\cS(\uppsi^{11})}{\cS(\uppsi^{21})}=
\frac{(\aly(m)-\du_{0,0})\uppsi^{21}}{(\du_{0,0}+\aly(m))\uppsi^{11}+\mc\uppsi^{21}}=
\frac{\aly(m)-\du_{0,0}}{(\du_{0,0}+\aly(m))\uppsi^{11}/{\uppsi^{21}}+\mc}=
\frac{\aly(m)-\du_{0,0}}{(\du_{0,0}+\aly(m))\mdu_{0,0}+\mc},\\
\label{dwp11}
\begin{split}
&\tso(\mdu)=\frac{\tso(\uppsi^{11})}{\tso(\uppsi^{21})}=
\frac{\dfrac{\mc (\du_{0,0}+\aly(m))}{\du_{0,1}-\aly(m+1)}\uppsi^{11}+2 \aly(m+1)\uppsi^{21}}%
{\dfrac{2 \aly(m+1) (\du_{0,0}+\aly(m))}{\aly(m+1)-\du_{0,1}}\uppsi^{11}+\mc\uppsi^{21}}=\\
&=\!\frac{\mc(\du_{0,0}\!+\!\aly(m))\uppsi^{11}/{\uppsi^{21}}\!+\!2\aly(m\!+\!1)(\du_{0,1}\!-\!\aly(m\!+\!1))}%
{-2\aly(m\!+\!1)(\du_{0,0}\!+\!\aly(m))\uppsi^{11}/{\uppsi^{21}}\!+\!\mc(\du_{0,1}\!-\!\aly(m\!+\!1))}
\!=\!\frac{\mc(\du_{0,0}\!+\!\aly(m))\mdu_{0,0}\!+\!2\aly(m\!+\!1)(\du_{0,1}\!-\!\aly(m\!+\!1))}%
{-2\aly(m\!+\!1)(\du_{0,0}\!+\!\aly(m))\mdu_{0,0}\!+\!\mc(\du_{0,1}\!-\!\aly(m\!+\!1))}.
\end{split}
\end{gather}
Since $\cS(\mdu)=\mdu_{1,0}$ and $\tso(\mdu)=\mdu_{0,1}$, 
from~\eqref{swp11},~\eqref{dwp11} one obtains
\begin{gather}
\label{swuu}
\mdu_{1,0}=\frac{\aly(m)-\du_{0,0}}{(\du_{0,0}+\aly(m))\mdu_{0,0}+\mc},\\
\label{dwuu}
\mdu_{0,1}=\frac{\mc(\du_{0,0}+\aly(m))\mdu_{0,0}+2\aly(m+1)(\du_{0,1}-\aly(m+1))}%
{-2\aly(m+1)(\du_{0,0}+\aly(m))\mdu_{0,0}+\mc(\du_{0,1}-\aly(m+1))}.
\end{gather}
Resolving~\er{swuu} with respect to $\du=\du_{0,0}$, we get
\begin{gather}
\label{mtly}
\du=\frac{\aly(m)-\mdu_{1,0}\big(\aly(m) \mdu_{0,0}+\mc\big)}{\mdu_{1,0}\mdu_{0,0}+1}.
\end{gather}
Now we substitute~\er{mtly} into~\er{dwuu},
using the relations $\du_{0,0}=\du$ and $\du_{0,1}=\tso(\du)$.
This gives the following equation
\begin{gather}
\label{eqbig}
\frac{\Big(\!\mdu_{1,1}\big(2 \aly(m) \mdu_{0,0} \mdu_{0,1}\!+\!\mc \mdu_{0,1}%
\!-\!2 \aly(m\!+\!1) \mdu_{0,0} \mdu_{1,0}\!-\!2 \aly(m\!+\!1)\big)\!+\!\mdu_{0,0}\big(2 \aly(m)\!-\!\mc \mdu_{1,0}\big)\!\Big)%
\big(2\aly(m\!+\!1)\mdu_{0,1}\!+\!\mc\big)}%
{\mc\big(\mdu_{1,1}(2\aly(m\!+\!1)\mdu_{0,1}\!+\!\mc\mdu_{0,0}\mdu_{1,0}\!+\!\mc)\!-\!2\aly(m\!+\!1)\mdu_{0,0}\mdu_{1,0})\!+\!4\aly(m)\aly(m\!+\!1)\mdu_{0,0}(\mdu_{0,1}\mdu_{1,1}\!+\!1\big)}=0.
\end{gather}

The numerator of the left-hand side of~\er{eqbig} is written as a product 
of two factors, the second one being $2\aly(m+1)\mdu_{0,1}+\mc$.
As shown in Remark~\ref{rcase} below, equation~\er{dwuu} 
and the regularity assumption from Remark~\ref{rreg} 
imply $2\aly(m+1)\mdu_{0,1}+\mc\neq 0$. Then \er{eqbig} yields
\begin{gather}
\label{eqmdu}
\mdu_{1,1}\big(2 \aly(m) \mdu_{0,0} \mdu_{0,1}+\mc \mdu_{0,1}%
-2 \aly(m+1) \mdu_{0,0} \mdu_{1,0}-2 \aly(m+1)\big)+\mdu_{0,0}\big(2 \aly(m)-\mc \mdu_{1,0}\big)=0.
\end{gather}

Formula~\er{mtly} gives an MT from equation~\er{eqmdu} to equation~\er{eqly},
in the following sense:\\
if, using $\du_{i,j}=\sop^i\tso^j(\du)$, $\,i,j\in\zz$, we substitute~\er{mtly} into~\er{eqly}, 
then we obtain an equation of the form $H(m,\mdu_{k,\ell},\ldots,\mc)=0$,
which holds modulo equation~\er{eqmdu} and its shifts.
(One can easily check this by a direct calculation, using~\er{eqmdu} and its shifts.)


\begin{remark}
\lb{rcase}
According to the regularity assumption from Remark~\ref{rreg},
we make computations in the domain where 
\begin{itemize}
	\item the denominators of the considered fractions are nonzero,
	\item the determinants of the matrices in the considered DLPs and GTs are nonzero.
\end{itemize}
Taking the denominators from~\er{dwuu},~\er{dethmk} and 
the determinant $\det\big(\widehat{\kd}^\mc\big)$ (computed in~\er{detmc}) 
of the matrix~$\widehat{\kd}^\mc$ from the DLP~\er{mcdlp}, we see that
the regularity assumption implies
\begin{gather}
\label{neq0}
\begin{gathered}
-2\aly(m+1)(\du_{0,0}+\aly(m))\mdu_{0,0}+\mc(\du_{0,1}-\aly(m+1))\neq 0,\\
\du_{0,1}-\aly(m+1)\neq 0,\qquad\quad \mc^2+4\aly(m+1)^2\neq 0.
\end{gathered}
\end{gather}
Using~\er{dwuu} and~\er{neq0}, we get
\begin{multline}
\notag
2\aly(m+1)\mdu_{0,1}+\mc=
2\aly(m+1)\frac{\mc(\du_{0,0}+\aly(m))\mdu_{0,0}+2\aly(m+1)(\du_{0,1}-\aly(m+1))}%
{-2\aly(m+1)(\du_{0,0}+\aly(m))\mdu_{0,0}+\mc(\du_{0,1}-\aly(m+1))}+\mc=\\
=\frac{\bigl(\mc^2+4\aly(m+1)^2\bigr)\bigl(\du_{0,1}-\aly(m+1)\bigr)}%
{-2\aly(m+1)(\du_{0,0}+\aly(m))\mdu_{0,0}+\mc(\du_{0,1}-\aly(m+1))}\neq 0.
\end{multline}
\end{remark}
\end{example}
\begin{example}
\lb{exkrmx}
S.~Konstantinou-Rizos, A.V.~Mikhailov, and 
P.~Xenitidis~\cite[Section~IV]{KMX2015} obtain the equation 
\begin{gather}
\lb{kmxp}
\frac{\pp_{1,1}-\pp_{1,0}}{\pp_{1,0} (\pp_{0,1} \pp_{1,1}-\pp_{1,0} \pp_{2,0})}-
\frac{\pp_{0,1}-\pp_{0,0}}{\pp_{0,1} (\pp_{-1,1} \pp_{0,1}-\pp_{0,0} \pp_{1,0})}=0
\end{gather}
for a scalar function $\pp=\pp(n,m)$ and present the DLP
\begin{gather}
\lb{lkmx}
\md=\begin{pmatrix}
{\la^2}/{\pp_{0,0}} & \la \\
 \la & 0
\end{pmatrix},\qquad\quad
\kd=\begin{pmatrix}
\la^2g+1 & \la\pp_{0,0}g \\
\la\pp_{-1,1}g & 1
\end{pmatrix},\\
\notag
g=\frac{\pp_{0,1}-\pp_{0,0}}{\pp_{0,0} (\pp_{-1,1} \pp_{0,1}-\pp_{0,0} \pp_{1,0})}.
\end{gather}
Here $\pp_{i,j}=\sop^i\tso^j(\pp)$, $\,i,j\in\zz$, can be regarded as dynamical variables.

Fix a nonzero constant $\ac\in\mathbb{C}$.
In what follows we use the matrices~\er{lkmx} with~$\la$ replaced by~$\ac$.
According to Definition~\ref{ddmlr}, since the matrices~\er{lkmx}
form a DLP for equation~\eqref{kmxp}, we can consider the auxiliary system
\begin{gather}
\label{kmxsy}
\cS(\Psi)=\begin{pmatrix}
{\ac^2}/{\pp_{0,0}} & \ac \\
 \ac & 0
\end{pmatrix}\cdot\Psi,\qquad\quad
\tso(\Psi)=\begin{pmatrix}
\ac^2g+1 & \ac\pp_{0,0}g \\
\ac\pp_{-1,1}g  &  1
\end{pmatrix}\cdot\Psi,
\end{gather}
which is compatible modulo equation~\eqref{kmxp}. 
Here $\Psi=\Psi(n,m)$ is a $2\times 2$ matrix
with the elements $\uppsi^{ij}=\uppsi^{ij}(n,m)$, $\,i,j=1,2$.

Substituting $\Psi=\begin{pmatrix}
\uppsi^{11}&\uppsi^{12}\\
\uppsi^{21}&\uppsi^{22}
\end{pmatrix}$ into~\eqref{kmxsy} and considering the first columns 
in the matrix equations~\eqref{kmxsy}, one gets
\begin{gather}
\label{sps11}
\cS(\uppsi^{11})=\frac{\ac^2}{\pp_{0,0}}\uppsi^{11}+\ac\uppsi^{21},\\
\label{sps21}
\cS(\uppsi^{21})=\ac\uppsi^{11},\\
\label{dps11}
\tso(\uppsi^{11})=(\ac^2g+1)\uppsi^{11}+\ac\pp_{0,0}g\uppsi^{21},\\
\label{dps21}
\tso(\uppsi^{21})=\ac\pp_{-1,1}g\uppsi^{11}+\uppsi^{21}.
\end{gather}
We set $\qq:=\uppsi^{21}$. Then, since $\cS(\uppsi^{21})=\qq_{1,0}$, 
from~\er{sps21} one obtains $\uppsi^{11}=\qq_{1,0}/\ac$.

Using $\uppsi^{21}=\qq=\qq_{0,0}$, $\,\uppsi^{11}=\qq_{1,0}/\ac$,
$\,\cS(\uppsi^{11})=\cS(\qq_{1,0}/\ac)=\qq_{2,0}/\ac$, $\,\pp=\pp_{0,0}$,
from~\er{sps11} we derive
\begin{gather}
\label{mtkmx}
\pp=\frac{\ac^2 \qq_{1,0}}{\qq_{2,0}-\ac^2 \qq_{0,0}},
\end{gather}
which implies 
\begin{gather}
\label{pm11}
\pp_{-1,1}=\cS^{-1}\tso\Big(\frac{\ac^2 \qq_{1,0}}{\qq_{2,0}-\ac^2 \qq_{0,0}}\Big)=
\frac{\ac^2\qq_{0,1}}{\qq_{1,1}-\ac^2 \qq_{-1,1}}.
\end{gather}
Equation~\er{dps21} yields
\begin{gather}
\label{gu}
g=\frac{\tso(\uppsi^{21})-\uppsi^{21}}{\ac\pp_{-1,1}\uppsi^{11}}.
\end{gather}
Substituting $\uppsi^{21}=\qq_{0,0}$, $\,\tso(\uppsi^{21})=\qq_{0,1}$, 
$\,\uppsi^{11}=\qq_{1,0}/\ac$, and~\er{pm11} into~\er{gu}, we get the relation
\begin{gather}
\label{gexp}
g=\frac{(\qq_{0,0}-\qq_{0,1}) (\ac^2 \qq_{-1,1}-\qq_{1,1})}{\ac^2\qq_{0,1}\qq_{1,0}}.
\end{gather}
Now we substitute $\uppsi^{21}=\qq_{0,0}$, $\,\uppsi^{11}=\qq_{1,0}/\ac$, 
$\tso(\uppsi^{11})=\qq_{1,1}/\ac$,  and~\er{mtkmx},~\er{gexp} into~\er{dps11}, which gives the equation
\begin{gather}
\label{qqeq}
\ac^2 \qq_{0,0} \qq_{0,1} \qq_{1,0}-\ac^2 \qq_{0,0} \qq_{0,1} \qq_{1,1}-\ac^2 \qq_{-1,1} \qq_{0,0} \qq_{2,0}
+\ac^2 \qq_{-1,1} \qq_{0,1} \qq_{2,0}-\qq_{0,1} \qq_{1,0} \qq_{2,0}+\qq_{0,0} \qq_{1,1} \qq_{2,0}=0.
\end{gather}

Formula~\er{mtkmx} provides an MT from equation~\er{qqeq}
to equation~\er{kmxp}, in the following sense:\\
if, using $\pp_{i,j}=\sop^i\tso^j(\pp)$, $\,i,j\in\zz$,
one substitutes~\er{mtkmx} into~\er{kmxp},
then one gets an equation of the form $F(\qq_{k,\ell},\ldots,\ac)=0$,
which holds modulo equation~\er{qqeq} and its shifts.
(It is straightforward to verify this, using~\er{qqeq} and its shifts.)


The invertible transformation
\begin{gather}
\notag
\qq\mapsto\ac^n\qq,\qquad\quad\pp\mapsto\ac\pp
\end{gather}
(which implies $\qq_{k,\ell}\mapsto\ac^{n+k}\qq_{k,\ell}$, $\,\pp_{k,\ell}\mapsto\ac\pp_{k,\ell}$)
preserves equation~\er{kmxp} and allows us to remove~$\ac$ from~\er{mtkmx},~\er{qqeq}.
Therefore, without loss of generality, one can replace~$\ac$ by~$1$ in~\er{mtkmx},~\er{qqeq}.
Thus we get the equation
\begin{gather}
\label{qqeq1}
\qq_{0,0} \qq_{0,1} \qq_{1,0}-\qq_{0,0} \qq_{0,1} \qq_{1,1}-\qq_{-1,1} \qq_{0,0} \qq_{2,0}
+\qq_{-1,1} \qq_{0,1} \qq_{2,0}-\qq_{0,1} \qq_{1,0} \qq_{2,0}+\qq_{0,0} \qq_{1,1} \qq_{2,0}=0
\end{gather}
and the MT
\begin{gather}
\label{mtkmx1}
\pp=\frac{\qq_{1,0}}{\qq_{2,0}-\qq_{0,0}}
\end{gather}
from~\er{qqeq1} to~\er{kmxp}. 
\end{example}
\begin{example}
\lb{exdlp}
The MT~\er{mtkmx1} from~\er{qqeq1} to~\er{kmxp}
and the DLP~\er{lkmx} for equation~\er{kmxp} 
allow us to derive a DLP for equation~\er{qqeq1} as follows.
Using $\pp_{i,j}=\sop^i\tso^j(\pp)$, we substitute~\er{mtkmx1} 
into the matrices~\er{lkmx}, and this gives the matrices
\begin{gather}
\label{dlpkm}
\overline{\md}(\qq_{0,0},\qq_{1,0},\qq_{2,0},\la)=\begin{pmatrix}
{\la^2(\qq_{2,0}-\qq_{0,0})}/{\qq_{1,0}} & \la \\
 \la & 0
\end{pmatrix},\qquad
\overline{\kd}=\begin{pmatrix}
\la^2\tilde{g}+1 & \la\hh_{0,0}\tilde{g} \\
\la\hh_{-1,1}\tilde{g} & 1
\end{pmatrix},\\
\notag
\tilde{g}=\frac{\hh_{0,1}-\hh_{0,0}}{\hh_{0,0} (\hh_{-1,1}\hh_{0,1}-\hh_{0,0}\hh_{1,0})},
\qquad\quad
\hh_{i,j}=\sop^i\tso^j\left(\frac{\qq_{1,0}}{\qq_{2,0}-\qq_{0,0}}\right),
\quad\qquad i,j\in\zz.
\end{gather}
Using equation~\er{qqeq1} and its shifts,
it is straightforward to verify that \er{dlpkm} is a DLP for~\er{qqeq1}.

Similarly to Example~\ref{exdgt},
let us try to simplify the DLP~\er{dlpkm} by applying a suitable gauge transformation~$\gbf$.
Namely, let us try to find an invertible $2\times 2$ matrix of the form $\gbf=\gbf(\qq_{0,0},\qq_{1,0},\la)$ 
such that in the corresponding gauge equivalent DLP
\begin{gather}
\lb{ovhmd}
\widehat{\overline{\md}}=\sop\big(\gbf(\qq_{0,0},\qq_{1,0},\la)\big)\cdot\overline{\md}(\qq_{0,0},\qq_{1,0},\qq_{2,0},\la)
\cdot\big(\gbf(\qq_{0,0},\qq_{1,0},\la)\big)^{-1},\\
\lb{ovhkd}
\widehat{\overline{\kd}}=\tso\big(\gbf(\qq_{0,0},\qq_{1,0},\la)\big)\cdot\overline{\kd}
\cdot\big(\gbf(\qq_{0,0},\qq_{1,0},\la)\big)^{-1}
\end{gather}
the matrix~\er{ovhmd} does not depend on~$\qq_{2,0}$.
Thus we require
\begin{gather}
\label{npdu11}
\dd{\qq_{2,0}}
\Big(\sop\big(\gbf(\qq_{0,0},\qq_{1,0},\la)\big)\cdot\overline{\md}(\qq_{0,0},\qq_{1,0},\qq_{2,0},\la)
\cdot\big(\gbf(\qq_{0,0},\qq_{1,0},\la)\big)^{-1}\Big)=0.
\end{gather}
Since $\big(\gbf(\qq_{0,0},\qq_{1,0},\la)\big)^{-1}$ does not depend on~$\qq_{2,0}$,
equation~\er{npdu11} is equivalent to $\dd{\qq_{2,0}}\big(\sop(\gbf)\cdot\overline{\md}\big)=0$,
which says that
\begin{gather}
\label{npdu11p}
\dd{\qq_{2,0}}
\Big(\sop\big(\gbf(\qq_{0,0},\qq_{1,0},\la)\big)\Big)\cdot\overline{\md}(\qq_{0,0},\qq_{1,0},\qq_{2,0},\la)+
\sop\big(\gbf(\qq_{0,0},\qq_{1,0},\la)\big)\cdot\dd{\qq_{2,0}}
\big(\overline{\md}(\qq_{0,0},\qq_{1,0},\qq_{2,0},\la)\big)=0.
\end{gather}
Multiplying~\er{npdu11p} by~$\overline{\md}(\qq_{0,0},\qq_{1,0},\qq_{2,0},\la)^{-1}$ from the right, 
one can rewrite~\er{npdu11p} as
\begin{gather}
\label{npdu11e}
\dd{\qq_{2,0}}\Big(\sop\big(\gbf(\qq_{0,0},\qq_{1,0},\la)\big)\Big)=
-\sop\big(\gbf(\qq_{0,0},\qq_{1,0},\la)\big)\cdot\dd{\qq_{2,0}}\big(\overline{\md}(\qq_{0,0},\qq_{1,0},\qq_{2,0},\la)\big)
\cdot\overline{\md}(\qq_{0,0},\qq_{1,0},\qq_{2,0},\la)^{-1}.
\end{gather}
Applying the operator~$\sop^{-1}$ to~\er{npdu11e}, we get the differential equation
\begin{gather}
\label{npdu11d}
\dd{\qq_{1,0}}\big(\gbf(\qq_{0,0},\qq_{1,0},\la)\big)=
-\gbf(\qq_{0,0},\qq_{1,0},\la)\cdot\sop^{-1}\Big(\dd{\qq_{2,0}}
\big(\overline{\md}(\qq_{0,0},\qq_{1,0},\qq_{2,0},\la)\big)
\cdot\overline{\md}(\qq_{0,0},\qq_{1,0},\qq_{2,0},\la)^{-1}\Big)
\end{gather}
for an unknown matrix~$\gbf(\qq_{0,0},\qq_{1,0},\la)$.
For the matrix $\overline{\md}(\qq_{0,0},\qq_{1,0},\qq_{2,0},\la)$ from~\er{dlpkm} one has
\begin{gather}
\label{npdu11mu}
\sop^{-1}\Big(\dd{\qq_{2,0}}\big(\overline{\md}(\qq_{0,0},\qq_{1,0},\qq_{2,0},\la)\big)
\cdot\overline{\md}(\qq_{0,0},\qq_{1,0},\qq_{2,0},\la)^{-1}\Big)=
\begin{pmatrix}
 0&{\la}/{\qq_{0,0}}\\
 0&0
 \end{pmatrix}.
\end{gather}
The elements of the $2\times 2$ matrix~$\gbf(\qq_{0,0},\qq_{1,0},\la)$ 
are denoted by~$g^{ij}=g^{ij}(\qq_{0,0},\qq_{1,0},\la)$, $\,i,j=1,2$.
Substituting~\er{npdu11mu} and 
$\gbf(\qq_{0,0},\qq_{1,0},\la)=\begin{pmatrix}
g^{11}(\qq_{0,0},\qq_{1,0},\la)&g^{12}(\qq_{0,0},\qq_{1,0},\la)\\
g^{21}(\qq_{0,0},\qq_{1,0},\la)&g^{22}(\qq_{0,0},\qq_{1,0},\la)
\end{pmatrix}$
into equation~\er{npdu11d}, we get
\begin{gather}
\label{npdu11d4}
\dd{\qq_{1,0}}\begin{pmatrix}
g^{11}(\qq_{0,0},\qq_{1,0},\la)&g^{12}(\qq_{0,0},\qq_{1,0},\la)\\
g^{21}(\qq_{0,0},\qq_{1,0},\la)&g^{22}(\qq_{0,0},\qq_{1,0},\la)
\end{pmatrix}=
-\begin{pmatrix}
g^{11}&g^{12}\\
g^{21}&g^{22}
\end{pmatrix}
\begin{pmatrix}
 0&{\la}/{\qq_{0,0}}\\
 0&0
 \end{pmatrix}.
\end{gather}

It is sufficient to find one (preferably as simple as possible) invertible matrix
$\gbf(\qq_{0,0},\qq_{1,0},\la)$ satisfying~\er{npdu11d4}. 
Equation~\er{npdu11d4} allows us to take $g^{21}=0$ and $g^{11}=g^{22}=1$.
Then \er{npdu11d4} reads $\dd{\qq_{1,0}}(g^{12})=-{\la}/{\qq_{0,0}}$, 
which allows one to set $\,g^{12}=-{\la}\qq_{1,0}/{\qq_{0,0}}$.
Thus we have constructed the matrix
\begin{gather}
\label{nmgu0}
\gbf(\qq_{0,0},\qq_{1,0},\la)=\begin{pmatrix}
1& -{\la}\qq_{1,0}/{\qq_{0,0}}\\
 0 & 1
\end{pmatrix},
\end{gather}
which can be used as a gauge transformation.
Applying the GT~\er{nmgu0} to the DLP~\er{dlpkm}
as written in~\er{ovhmd},~\er{ovhkd}, we get the corresponding gauge equivalent DLP
\begin{gather}
\lb{wovmk}
\widehat{\overline{\md}}=
\begin{pmatrix}
 -{\la^2\qq_{0,0}}/{\qq_{1,0}}&\la(1-\la^2)\\
 \la&{\la^2\qq_{1,0}}/{\qq_{0,0}}
 \end{pmatrix},\qquad
\widehat{\overline{\kd}}=
\begin{pmatrix}
  1-\dfrac{\la^2\qq_{-1,1}(\qq_{0,1}-\qq_{0,0})}{\qq_{0,1}\qq_{1,0}}
 &
 \dfrac{\la(1-\la^2)\qq_{-1,1}(\qq_{0,1}-\qq_{0,0})}
 {\qq_{0,0}\qq_{0,1}}
 \\
  \dfrac{\la(\qq_{0,1}-\qq_{0,0})}{\qq_{1,0}}
 & 1+\dfrac{\la^2(\qq_{0,1}-\qq_{0,0})}{\qq_{0,0}}
 \end{pmatrix}.
\end{gather}
The computation of~\er{wovmk} by formulas~\er{ovhmd},~\er{ovhkd} has been done
modulo equation~\er{qqeq1} and its shifts, since we work with DLPs for equation~\er{qqeq1}.
(Namely, we have used~\er{qqeq1} and its shifts, in order to simplify the matrix~$\widehat{\overline{\kd}}$.) 

The obtained DLP~\er{wovmk} is simpler than~\er{dlpkm}.
In particular, $\widehat{\overline{\md}}$ in~\er{wovmk}
depends only on~$\qq_{0,0},\,\qq_{1,0},\,\la$, in contrast to
$\overline{\md}$ in~\er{dlpkm} depending on~$\qq_{0,0},\,\qq_{1,0},\,\qq_{2,0},\,\la$.
\end{example}

\section*{Acknowledgments}

The author would like to thank S.~Konstantinou-Rizos and A.V.~Mikhailov for useful discussions.

This work was supported by the Russian Science Foundation (grant No. 25-21-00454, 
\url{https://rscf.ru/project/25-21-00454/} ).

\end{document}